\documentclass{lmcs}
\pdfoutput=1

\usepackage{lastpage}
\lmcsdoi{15}{4}{1}
\lmcsheading{}{\pageref{LastPage}}{}{}%
{Mar.~09,~2018}{Oct.~04,~2019}{}


\input{macros.tex}

\begin{document}

\newcommand{\dedication}{\mbox{In memory of Zolt\'an
    \'Esik}}

\title[\texorpdfstring{$\aleph_{1}$ and the modal $\mu$-calculus}{Aleph1 and the modal mu-calculus}]%
  {\texorpdfstring{$\aleph_{1}$ and the modal $\mu$-calculus\\[2mm] \dedication}{Aleph1 and the model mu-calculus, \dedication}}

\author[M. J. Gouveia]{
  Maria Jo\~{a}o Gouveia
}
\address{
  Faculdade de Ci\^{e}ncias, Universidade de Lisboa, Portugal
}
\email{mjgouveia@fc.ul.pt}
\thanks{Partially supported by FCT under grant
  SFRH/BSAB/128039/2016.}

\author[L. Santocanale]{Luigi Santocanale}
\address{
  LIS, CNRS UMR 7020, Aix-Marseille Universit\'e, France}
\email{luigi.santocanale@lis-lab.fr}


\begin{abstract}
  For a regular cardinal $\kappa$, a formula of the modal
  $\mu$-calculus is \continuous in a variable $x$ if, on every model,
  its interpretation as a unary function of $x$ is monotone and
  preserves unions of \dirset{s}. We define the fragment $\C(x)$ of
  the modal $\mu$-calculus and prove that all the formulas in this
  fragment are \continuous[$\aleph_{1}$].
  For each formula $\phi(x)$ of the modal $\mu$-calculus, we construct
  a formula $\psi(x) \in \C(x)$ such that $\phi(x)$ is \continuous,
  for some $\kappa$, if and only if $\phi(x)$ is equivalent to
  $\psi(x)$. Consequently, we prove that (i) the problem whether a
  formula is \continuous for some $\kappa$ is decidable, (ii) up to
  equivalence, there are only two fragments determined by continuity
  at some regular cardinal: the fragment $\C[0](x)$ studied by
  Fontaine and the fragment $\C(x)$.
  We apply our considerations to the problem of characterizing closure
  ordinals of formulas of the modal $\mu$-calculus.
  An ordinal $\alpha$ is the closure ordinal of a formula $\phi(x)$ if
  its interpretation on every model converges to its least fixed-point
  in at most $\alpha$ steps and if there is a model where the
  convergence occurs exactly in $\alpha$ steps.
  We prove that $\omega_{1}$, the least uncountable ordinal, is such a
  closure ordinal. Moreover, we prove that closure ordinals are closed
  under ordinal sum. Thus, any formal expression built from
  $0,1,\omega,\omega_{1}$ by using the binary operator symbol $+$ gives
  rise to a closure ordinal.
\end{abstract}


\maketitle

\section{Introduction}

The propositional modal $\mu$-calculus \cite{Kozen83,LenziIntro} is a
well established logic in theoretical computer science, mainly due to
its convenient properties for the verification of computational
systems.  It includes as fragments many other computational logics,
PDL, CTL, CTL${}^{\ast}$, its expressive power is therefore highly
appreciated.  
Also, being capable to express all the bisimulation invariant
properties of transition systems that are definable in monadic second
order logic, the modal $\mu$-calculus can itself be considered as a
robust fragment of an already very expressive logic \cite{janwal}.
Despite its strong expressive power, this logic is still considered as
a tractable one. 
\correction{Its model checking problem, even if in the class
  $\textrm{UP} \cap \textrm{co-UP}$ \cite{jurdzinski}, becomes
  polynomial as soon as some critical parameters are fixed or
  restricted classes of models are considered
  \cite{Obdrzalek2003,alberucci09,bojanczyk}.}{
  Its model checking problem, known to be in the class
  $\textrm{UP} \cap \textrm{co-UP}$ \cite{jurdzinski}, has recently
  been proved to be quasi-polynomial and fixed-parameter tractable
  \cite{CaludeJKL017}. Moreover, this problem becomes polynomial if
  some restricted classes of models are considered
  \cite{Obdrzalek2003,alberucci09,bojanczyk}.}
The widespread interest for this logic has triggered further
researches that spread beyond the realm of verification: these concern
the expressive power \cite{bradfield,BGL2002}, axiomatic bases
\cite{walukiewicz}, algebraic and order theoretic approaches
\cite{completions}, deductive systems \cite{niwinski,Studer2008}, and
the semantics of functional programs \cite{Fortier}.

The present paper lies at the intersection of two lines of research on
the modal $\mu$-calculus, on continuity \cite{Fontaine08} and on
closure ordinals \cite{Czarnecki,AfshariL13}.  Continuity of monotone
functions is a fundamental phenomenon in modal logic, on which
well-known uniform completeness theorems rely
\cite{Sahlqvist,Sambin,Jonsson}.  Fontaine~\cite{Fontaine08}
characterized the formulas of the modal $\mu$-calculus that give rise
to continuous functions on Kripke models.  It is well-known, for
example in categorical approaches to model theory \cite{adamek1994},
that the notion of continuity of monotone functions (and of functors)
can be generalized to \kcontinuity, where the parameter $\kappa$ is an
infinite regular cardinal.
In the work \cite{ITA2002} one of the authors proved that
\continuous[$\aleph_{1}$] functors are closed under their greatest
fixed-points.  Guided by this result, we present in this paper a
natural syntactic fragment $\C(x)$ of the modal $\mu$-calculus whose
formulas are \continuous[$\aleph_{1}$]---that is, they give rise to
\continuous[$\aleph_{1}$] monotone unary functions of the variable $x$
on arbitrary models. A first result that we present here is that
\emph{the fragment $\C(x)$ is decidable}: for each $\phi(x) \in \Lmu$,
we construct a formula $\psi(x) \in \C(x)$ such that $\phi(x)$ is
\continuous[$\aleph_{1}$] on every model if and only if $\phi(x)$ and
$\psi(x)$ are semantically equivalent formulas.  We borrow some
techniques from \cite{Fontaine08}, yet the construction of the formula
$\psi(x)$ relies on a new notion of normal form for formulas of the
modal $\mu$-calculus. A closer inspection of our proof uncovers a
stronger fact: the formulas $\phi(x)$ and $\psi(x)$ are equivalent if
and only if, for some regular cardinal $\kappa$, $\phi(x)$ is
\continuous on every model. The stronger statement implies that we
cannot find a fragment $\CC[\kappa](x)$ of \continuous formulas for
some cardinal $\kappa$ strictly larger than $\aleph_{1}$; any such
hypothetical fragment collapses, semantically, to the fragment
$\C(x)$. 
\correction{Since our conference paper \cite{GS17} we were informed
  that Fontaine's work \cite{Fontaine08} had been extended to a
  preprint that finally appeared in print \cite{FontaineVenema2018}.
  In the extended version the fragment $\C(x)$ is also studied, yet
  the semantic property pinpointed there and corresponding to the
  syntactic fragment $\C(x)$ is different from the property that we
  consider, \kcontinuity.}{
  In \cite{FontaineVenema2018}, an extended journal version of the
  conference paper \cite{Fontaine08}, the fragment $\C(x)$ is also
  studied, yet the semantic property pinpointed there and
  corresponding to the syntactic fragment $\C(x)$ is different from
  the property that we consider, \kcontinuity.  } Say that a formula
of the modal $\mu$-calculus has the \emph{finite width property} if,
whenever it is satisfied in a tree model, it is satisfied in a
finitely branching subtree of this model.  It is proved in
\cite{FontaineVenema2018} that a formula has the finite width property
if and only if it is equivalent to a formula in $\C(x)$.  Combining
these results with ours, we deduce a quite surprising statement: a
formula has the finite width property if and only if it is
\kcontinuous for some regular cardinal $\kappa$. While it is easy to
guess why the finite width property implies
\kcontinuity[$\aleph_{1}$], the \correction{other direction of the
  equivalence just stated}{converse implication} appears to be a
non-obvious strong statement, whose potential consequences and
applications need to be uncovered.

Our interest in \continuity[$\aleph_{1}$] was wakened once more when
researchers started investigating closure ordinals of formulas of the
modal $\mu$-calculus \cite{Czarnecki,AfshariL13}.  The notion of
closure ordinal was studied in the context of first order inductive
definitions~\cite{EIAS}. Closure ordinals for the modal $\mu$-calculus
are more directly related to \emph{global} inductive definability, see
\cite{BM1978}, in that a class of structures is being tested, not a
single structure. We consider closure ordinals as a wide field where
the notion of \kcontinuity can be exemplified and applied; the two
notions---\kcontinuity and closure ordinals---are, in our opinion,
naturally intertwined and the results we present in this paper are in
support of this thesis.
An ordinal $\alpha$ is the closure ordinal of a formula $\phi(x)$ if
(the interpretation of) this formula (as a monotone unary function of
the variable $x$) converges to its least fixed-point $\mu_{x}.\phi(x)$
in at most $\alpha$ steps in every model and, moreover, there exists
at least one model in which the formula converges exactly in $\alpha$
steps.  Not every formula has a closure ordinal. For example, the
simple formula $\nec[\;] x$ has no closure ordinal; more can be said,
this formula is not \kcontinuous for any $\kappa$.  As a matter of
fact, if a formula $\phi(x)$ is \kcontinuous (that is, if its
interpretation on every model is \kcontinuous), then it has a closure
ordinal $\clor(\phi(x))\leqslant \kappa$---here we use the fact that,
using the axiom of choice, a cardinal can be identified with a
particular ordinal, for instance $\aleph_{0} = \omega$ and
$\aleph_{1} = \omega_{1}$. Our results on \continuity[$\aleph_{1}$]
show that all the formulas in $\C(x)$ have a closure ordinal bounded
by $\omega_{1}$.
For closure ordinals, our results are threefold. Firstly we prove that
\emph{the least uncountable ordinal $\omega_1$ belongs to the set
  $\OrdLmu$} of all closure ordinals of formulas of the propositional
modal $\mu$-calculus.
Secondly, we prove that \emph{$\OrdLmu$ is closed under ordinal
  sum}. It readily follows that any formal expression built from
$0,1,\omega,\omega_1$
by using the binary operator symbol $+$ gives rise to an ordinal in
$\OrdLmu$.  Let us recall that \Czarnecki \cite{Czarnecki} proved that
all the ordinals $\alpha < \omega^{2}$ belong to $\OrdLmu$. Our
results generalize \Czarnecki's construction of closure ordinals and
give it a rational reconstruction---every ordinal strictly smaller
than $\omega^{2}$ can be generated by $0,1$ and $\omega$ by repeatedly
using the sum operation.
Finally, even considering that our work does not yield methods to
exclude ordinals from $\OrdLmu$, the fact that there are no relevant
fragments of the modal $\mu$-calculus determined by continuity at some
regular cardinal other than $\aleph_{0}$ and $\aleph_{1}$ implies that
\emph{the methodology} (adding regular cardinals to $\OrdLmu$ and
closing them under ordinal sum) \emph{used until now to construct new
  closure ordinals} for the modal $\mu$-calculus \emph{cannot be
  further exploited}.

Let us add some final considerations.  The fragment $\C(x)$ of the
propositional modal \mbox{$\mu$-calculus} has imposed itself by its
robustness, which can be recognised in our work as well as in
\cite{FontaineVenema2018}.  We believe $\C(x)$ is worth investigating
further in order to enlighten a hidden dimension (and thus new tools,
new ideas, new perspectives, etc.) of the modal $\mu$-calculus and of
fixed-point logics. As an example, take the modal $\mu$-calculus on
deterministic models: states have at most one successor and it is
immediate to conclude that every formula is \continuous[$\aleph_{1}$]
on these models. Whether this and other observations can be exploited
(towards understanding alternation hierarchies or reasoning using
axiomatic bases, for example) is part of future research.
We also believe that the scope of this work, as well as of the
problems studied within, goes much beyond the pure theory of the modal
$\mu$-calculus.
  For example, our interest in closure ordinals stems from a previous
  proof-theoretic investigation of induction and coinduction
  \cite{Fortier,ITA2002}.
  In these works ordinal notations are banned from the syntax
    because of an alleged non-constructiveness of the set theory
    needed to represent ordinals.   However, also considering
  that elegant constructive theories of ordinals exist, see
  e.g. \cite{AST}, the present work encourages us to develop
  alternative proof-theoretic frameworks based on ordinals.

The paper is structured as follows. In Section~\ref{sec:background} we
introduce the notion of \kcontinuity. In the following
Section~\ref{sec:Kleene} we illustrate the interactions between
\kcontinuity and least/\GFP{s} of monotone maps.  In
Section~\ref{sec:mucalculus} we present the modal $\mu$-calculus and
some of the related theory that we shall need in the following
sections. Section~\ref{sec:continuousfragment} presents our results on
the fragment $\C(x)$.  The following Section~\ref{sec:models} presents
a tool---roughly speaking the observation that various kind of
submodels can be logically described modulo the introduction of a new
propositional variable---that is repeatedly used in the rest of the
paper to obtain results on closure ordinals. In
Section~\ref{sec:closureordinals} we argue that the least uncountable
ordinal is a closure ordinal for the modal $\mu$-calculus. In the
final Section~\ref{sec:ordinalsum} we argue that $\OrdLmu$, the set of
closure ordinal of formulas of the modal $\mu$-calculus, is closed
under ordinal sum.

\bigskip

{\bf Acknowledgment.} The authors are thankful to the anonymous
reviewers for their analytical reading and for their valuable comments
by which
the presentation of this research could substantially improve.


\tableofcontents

\section{\continuous maps}
\label{sec:background}

In this section we consider \continuity of monotone maps between
powerset Boolean algebras, where the parameter $\kappa$ is an infinite
regular cardinal. If $\kappa = \aleph_{0}$, then \continuity coincides
with the usual notion of continuity as found for example in
\cite{Fontaine08,FontaineVenema2018}.  
The reader might find further information in the monograph
\cite{adamek1994} where this notion is presented in the more general
context of categories.

\medskip

In the following $\kappa$ is an infinite regular cardinal, $A$ and $B$
are sets, for which $P(A)$ and $P(B)$ denote the corresponding
powerset Boolean algebras, and $f : P(A) \rto P(B)$ is a monotone map.
We shall say that a subset $X$ of a set $A$ is \emph{\ksmall} if
$\card X < \kappa$. For example, a set $X$ is \ksmall[$\aleph_{0}$] if
and only if it is finite, and it is \ksmall[$\aleph_{1}$] if and only
if it is countable.  Regularity of the cardinal $\kappa$ essentially
amounts to the following property: if ${\cal J}$ is a \ksmall
collection of \ksmall subsets of $A$, then $\bigcup {\cal J}$ is
\ksmall.

\begin{defi}
  A subset ${\cal I} \subseteq P(A)$ is a \emph{\dirset} if
  every collection ${\cal J} \subseteq {\cal I}$ with
  $\card {\cal J} < \kappa$ has an upper bound in ${\cal I}$.
  \label{def:kappacontinuity}
  A map $f : P(A) \rto P(B)$ is \emph{\continuous} if
  $f(\bigcup {\cal I}) = \bigcup f({\cal I})$, whenever ${\cal I}\subseteq P(A)$ is
  a \dirset.
\end{defi}
  
Observe that if $\kappa'$ is a regular cardinal and
$\kappa < \kappa'$, then a $\kappa'$-directed set is also a
\dirset. Therefore, if $f$ is \continuous, then it also preserves
unions of $\kappa'$-directed sets, thus it is $\kappa'$-continuous as
well.  Also, notice that the wording ``monotone \kcontinuous'' is
redundant: if $f$ is \kcontinuous, then it is monotone, since if
$X \subseteq Y$, then $\set{X,Y}$ is \kdirected, so
$f(Y) = f(Y \cup X) = f(X) \cup f(Y)$, so $f(X) \subseteq f(Y)$.

For each subset $X$ of $A$, define
\begin{align*}
  \Ik(X) & := \set{ X' \mid X' \subseteq X, X' \text{ is \ksmall}}\,.
\end{align*}
Notice that $\bigcup \Ik(X) = X$ and $\Ik(X)$ is a \dirset. For this
latter property, it is useful to note that if
$\set{X_{i}\subseteq X \mid i \in I}$ is a \ksmall set of \ksmall
subsets of $X$, then the union $\bigcup \set{X_{i}\mid i \in I}$ is
still \ksmall, so it belongs to $\Ik(X)$.

\begin{prop}
  \label{prop:ksmall}
  A subset $X$ of $A$ 
  is \ksmall if and only if, for every \dirset ${\cal I}$,
  $X \subseteq \bigcup {\cal I}$ implies $X \subseteq I$ for some
  $I \in {\cal I}$.
\end{prop}
\begin{proof}
  We firstly prove that if $X$ is \ksmall and
  ${\cal I} \subseteq P(A)$ is a \dirset such that
  $X \subseteq \bigcup {\cal I}$, then there exists $I \in {\cal I}$
  with $X \subseteq I$.  For each $a \in X$, let $I_{a} \in {\cal I}$
  such that $a \in I_{a}$. Then ${\cal J} = \set{I_{a}\mid a \in X}$
  is a subfamily of ${\cal I}$ with $\card {\cal J} < \kappa$, whence
  there exists $I \in {\cal I}$ with $I_{a} \subseteq I$, for each
  $a \in X$; whence $X \subseteq I$.

  For the converse, recall that $X = \bigcup\mathcal{I}_{\kappa}(X)$
  and that $\mathcal{I}_{\kappa}(X)$ is a \dirset.  Suppose therefore
  that, for every \dirset ${\cal I}$, $X \subseteq \bigcup {\cal I}$
  implies $X \subseteq I$ for some $I \in \mathcal{I}$. Applying this
  property when ${\cal I} =
  \mathcal{I}_{\kappa}(X)$ 
  yields $X \subseteq X'$ for some \ksmall $X' \subseteq X$.
  Therefore $X' = X$ and $X$ is \ksmall.
\end{proof}

\begin{prop}
  \label{prop:2characcont}
  A monotone map $f : P(A) \rto P(B)$ is $\kappa$-continuous if
  and only if, for every $X \in P(A)$,
\begin{align*}
    f(X)& =\bigcup\set{f(X')\mid X'\subseteq X,
      X' \text{ is \ksmall} }\,.
\end{align*}
\end{prop}
\begin{proof}
  Let $f : P(A) \rto P(B)$ be a \kcontinuous monotone map. Notice
  that the equation above is $f(\bigcup \Ik(X)) = \bigcup f(\Ik(X))$,
  since $X = \bigcup \Ik(X)$. The equation holds since $\Ik(X)$ is
  \kdirected and we are supposing that $f$ is \kcontinuous.

  Conversely suppose that $f : P(A) \rto
  P(B)$ is a monotone map such that $f(X)=\bigcup
  f({\Ik(X)})$ for every $X \in P(A)$. 
  Also let 
  ${\cal I\subseteq
    P(A)}$ be a \kdirset, so we aim to show that $f(\bigcup {\cal I})
  = \bigcup f({\cal I})$. \correction{By hypothesis}{Since
    $f$ is \continuous,} $f(\bigcup {\cal I}) = \,\bigcup
    \,f({\Ik(\bigcup {\cal I})})\,$.  Since
    $f$ is monotone, we have $\bigcup f({\cal I})\subseteq f(\bigcup
    {\cal
      I})$ and therefore we only need to verify the opposite
    inclusion.
  Let $Y$ be a \ksmall set contained in $\bigcup {\cal
    I}$.  By Proposition~\ref{prop:ksmall} there exists $Z \in {\cal
    I}$ such that $Y \subseteq
  Z$. 
  Hence for every $Y \in \Ik(\bigcup {\cal I})$ there exists $Z\in
  \cal I$ such that $Y\subseteq Z$ and so also $f(Y)\subseteq
  f(Z)$. Thus, $\bigcup \,{f(\Ik(\bigcup {\cal I}))}\subseteq \bigcup
  {f(\cal I)}$.
  Consequently, we have
  \begin{align*}
    f(\bigcup {\cal I}) =\,\bigcup \,f({\Ik(\bigcup {\cal I})})\,\subseteq
    \bigcup f({\cal I})\,,
 \end{align*}
  proving the opposite inclusion.
\end{proof}

Next we extend the notion of \kcontinuity to functions of many
variables, that is, to functions whose domain is a finite product of
the form $P(A_{1}) \times \ldots \times P(A_{n})$, the ordering being
coordinate-wise.  To achieve this goal, we observe that there is a
standard isomorphism
$\psi : P(A_{1} \dunion \ldots \dunion A_{n}) \rto P(A_{1}) \times
\ldots \times P(A_{n})$, where $\dunion$ denotes the disjoint union.
Therefore, we say that a monotone function
$f : P(A_{1}) \times \ldots \times P(A_{n}) \rto P(B)$ is \kcontinuous
if the function of one variable
$f \circ \psi : P(A_{1} \dunion \ldots \dunion A_{n}) \rto P(B)$ is
\kcontinuous.
The standard isomorphism associates to a subset
$S \subseteq A_{1} \dunion \ldots \dunion A_{n}$ the tuple
$\psi(S) = \langle S \cap A_{1},\ldots ,S \cap A_{n}\rangle$.
The next Lemma (that, for simplicity, we state and prove for $n = 2$)
states the expected property of \kcontinuous functions of many
variables: these functions are \kcontinuous exactly when they are
\kcontinuous in each variable.
\begin{lem}
  \label{lemma:bicontinuity}
  A monotone map $f : P(A_{1})\times P(A_{2}) \rto P(B)$ is
  \kcontinuous w.r.t. the coordinate-wise order on
  $P(A_{1})\times P(A_{2})$ if and only if it is \kcontinuous in every
  variable.
\end{lem}
\begin{proof}
  Obviously if 
  $f\circ \psi\colon P(A_1 \dunion A_2) \rto P(B)$ is \kcontinuous,
  then it is \kcontinuous when we fix a subset,
  say $X \subseteq A_1$. Indeed, a family of the form
  $\set{X \dunion Y_{i} \mid i \in I,\, Y_{i }\subseteq A_{2}} $ is
  \kdirected if and only if $\set{Y_{i} \subseteq A_2\mid i \in I } $
  is \kdirected.

Conversely, suppose that
$f\circ \psi\colon P(A_1 \dunion A_2) \rto P(B)$ is \kcontinuous in
every variable.  First observe that, for any families
$\mathcal X=\set{X_{i} \subseteq A_1\mid i \in I }$ and
$\mathcal Y=\set{Y_{i} \subseteq A_2\mid i \in I } $, we have
that
\[\bigcup_{i}\set{X_i \dunion Y_{i} \mid i \in I} =
  \bigcup_{i,j}\set{X_i \dunion Y_{j} \mid i,j \in I}
  =\bigcup_{i}(X_{i} \dunion \bigcup_{j} Y_{j})\] and when
$\set{X_i \dunion Y_{i} \mid i \in I} $ is \kdirected also
$\mathcal X$, $\mathcal Y$ and
$\set{X_i\dunion Y_{j} \mid i,j \in I } $ are \kdirected.
Consequently, given a \kdirset
$\set{X_{i} \dunion Y_{i} \mid i \in I}$ with $X_{i} \subseteq A_1$
and $Y_{i } \subseteq A_2$, the following holds
\begin{align*}
  (f\circ \psi)(\bigcup_{i} X_{i} \dunion Y_{i}) & = (f\circ
  \psi)(\bigcup_{i,j} X_{i} \dunion Y_{j}) 
  = (f\circ
  \psi)(\bigcup_{i} (X_{i} \dunion \bigcup_{j} Y_{j})) \\
  & =\bigcup_{i} (f\circ \psi)(X_{i} \dunion \bigcup_{j} Y_{j})
  =\bigcup_{i} \bigcup_{j} (f\circ \psi)(X_{i} \dunion Y_{j})\,,
  \tag*{since $f$ is \continuous in each
    variable,}\\
  &= \bigcup_{i,j} (f\circ \psi)(X_{i} \dunion Y_{j}) = \bigcup_{i}
  (f\circ \psi)(X_{i} \dunion Y_{i})\,.
\end{align*}
 \ENDOFPROOF[Lemma~\ref{lemma:bicontinuity}]
\end{proof}


\section{Fixed-points of \kcontinuous maps} 
\label{sec:Kleene}
The interplay between \continuity of monotone maps (recall that
$\kappa$ is assumed to be an infinite regular cardinal) and their
least and greatest \FP{s} is the focus of the present section. On the
one hand, the Knaster-Tarski theorem \cite{tars:deci55} states that
the \LFP of a monotone map $f : P(A) \rto P(A)$ is the set
$\bigcap \set{ X \subseteq A \mid f(X) \subseteq X}$.  On the other
hand, Kleene's fixed-point theorem states that the \LFP of an
\continuous[$\aleph_{0}$] map $f$ is constructible by iterating
$\omega$-times $f$ starting from the empty set, namely it is equal
to $\bigcup_{n \geq 0} f^{n}(\emptyset)$.  Generalisations of Kleene's
theorem appeared later and give ways to build the \LFP of monotone
maps by ordinal approximations; see \cite{LassezNS82} for an
historical account of this family of theorems.

The first result we present in this section is a generalised Kleene's
fixed-point theorem specifically suited to \continuous maps (
we do not claim the authorship of Proposition~\ref{prop:convkcontinuous},
even if we could not find it stated as it is in the literature).

\begin{defi}
  \label{def:approximants}
  Let $f : P(A) \rto P(A)$ be a monotone map. The
  \emph{approximants} $f^{\alpha}(\emptyset)$, with $\alpha$ an ordinal, 
    are inductively defined as follows:
  \begin{align*}
    f^{\alpha + 1}(\emptyset) & := f(f^{\alpha}(\emptyset))\,,
    \qquad
    f^{\alpha}(\emptyset)  := \textstyle{\bigcup_{\beta < \alpha}}\,
    f^{\beta}(\emptyset)\,
    \quad\text{when $\alpha$ is a limit ordinal.}
  \end{align*}
  We say that \emph{$f$ converges to its least fixed-point in at most
    $\alpha$ steps} if $f^{\alpha}(\emptyset)$ is a fixed-point
  (necessarily the least one) of $f$.  We say that \emph{$f$ converges
    to its least fixed-point in exactly $\alpha$ steps} if
  $f^{\alpha}(\emptyset)$ is a fixed-point of $f$ and
    $f^{\beta}(\emptyset)\subsetneq f^{\beta + 1}(\emptyset)$, for
    each ordinal $\beta < \alpha$.
  \end{defi}
  Let us recall that in set theory a cardinal $\kappa$ is identified with the
  least ordinal of cardinality equal to $\kappa$.  We exploit this,
  notationally, in the next proposition.
\begin{prop}
  \label{prop:convkcontinuous}
  If $f : P(A) \rto P(A)$ is a \continuous monotone map, then $f$
  converges to its least fixed-point in at most $\kappa$ steps.
\end{prop}
\begin{proof}
 Let us argue that $f^{\kappa}(\emptyset)$ is a \correction{fixed-point}{prefixed-point}
  of $f$:
  \begin{align*}
    f(\,f^{\kappa}(\emptyset)\,) & = f(\,\bigcup_{\alpha <
      \kappa}f^{\alpha}(\emptyset)\,) 
     = \bigcup_{\alpha < \kappa} f(f^{\alpha}(\emptyset))\,
    \subseteq \bigcup_{\alpha < \kappa}
    f^{\alpha}(\emptyset) = f^{\kappa}(\emptyset)\,
  \end{align*}
  since the regularity of $\kappa$ implies that
  $\set{ f^{\alpha}(\emptyset) \mid \alpha < \kappa}$ is a
  \dirset. Since the inclusion
  $f^{\kappa}(\emptyset) \subseteq f(f^{\kappa}(\emptyset))$ holds by
  monotonicity of $f$, $f^{\kappa}(\emptyset)$ is also a fixed-point
  of $f$.
\end{proof}
Until now we have focused on \LFP{s} of monotone maps.  \GGFP{s}
  are dual to \LFP{s}: namely, for a monotone map
  $f : P(A) \rto P(A)$, its \GFP is the largest subset $Z$ of $A$ such
  that $f(Z) = Z$; by Tarski's theorem, it is equal to
  $\bigcup \set{Z \subseteq A \mid Z \subseteq f(Z)}$.
  Propositions~\ref{prop:nucontinuous} and 
  \ref{prop:mucontinuous} relate both kind of (parametrized) fixed
  points to continuity; they are specific instances of a result stated
  for categories in \cite{ITA2002}.  To clarify their statements, let
  us recall that if $f : P(B)\times P(A) \rto P(B)$ is a monotone map,
  then, for each $X\in P(A)$, 
the unary map $f(-, X) : P(B) \rto P(B)$, $Z \mapsto f(Z,X)$, is also
monotone. Hence, we may consider the map $P(A) \rto P(A)$ that sends
$X$ to the least (resp. greatest) fixed-point of $f(-, X)$; by using
the standard $\mu$-calculus notation, we denote it by $\mu_{z}.f(z,-)$
(resp. $\nu_{z}.f(z,-)$).\footnote{
  Let us mention that later we shall emphasize the distinction
  syntax/semantics. Then, we shall use $\lfp$ and $\gfp$ in the semantics
  for the symbols $\mu$ and $\nu$, respectively, and reserve these
  symbols for the syntax.  } Let us also recall that $f$ is
\continuous w.r.t. the coordinatewise order on $P(B)\times P(A)$ if
and only if it is \continuous in every variable (see
Lemma~\ref{lemma:bicontinuity}\,).

\begin{prop}
  \label{prop:nucontinuous}
  Let $f : P(B)\times P(A) \rto P(B)$ be a \continuous monotone
  map.   If $\kappa > \aleph_{0}$ then
  $\nu_{z}.f(z,\,_{-}) : P(A) \rto P(B)$ is also \continuous.
\end{prop}
\begin{proof}
  Let us write $g(x) := \nu_{z}.f(z,x)$.  We shall show that,
    for every $b\in B$ and for every $X\in P(A)$, if $b \in g(X)$, then
    $b \in g(X')$ for some \ksmall $X'$ contained in $X$.  Having
  shown this, 
  the continuity of $g$ follows from
  Proposition~\ref{prop:2characcont}.
  \correction{ Note that the condition $b \in g(X)$ holds when there
    exists $Z \subseteq B$ such that $b \in Z$ and
    $Z \subseteq f(Z,X)$. Aiming \correction{to find}{at constructing}
    such a set $Z$, we recursively \correction{obtain}{define} a
    family $(X_n)_{n \geq 1}$ of \ksmall subsets of $X$ and a family
    $(Z_n)_{n \geq 0}$ of \ksmall subsets of \correction{$Z$}{$B$}
    satisfying $Z_n\subseteq f(Z_{n+1},X_{n+1})$.}{
    Let therefore $b \in g(X)$ and note that this condition implies
    that, for some $Z \subseteq B$, $b \in Z$ and
    $Z \subseteq f(Z,X)$; let us fix such $Z$. Aiming at constructing
    a \ksmall subset $X' \subseteq A$ such that $b \in g(X')$, we
    recursively define a family $(X_n)_{n \geq 1}$ of \ksmall subsets
    of $X$ and a family $(Z_n)_{n \geq 0}$ of \ksmall subsets of $Z$
    satisfying $Z_n\subseteq f(Z_{n+1},X_{n+1})$.  }

  For $n=0$ we take $Z_{0} := \set{b}$ which is a \ksmall subset of $f(Z,X)$.
  Now suppose we have already constructed
  a \ksmall set $Z_{n}$ that satisfies $Z_{n} \subseteq f(Z,X)$.
  Let us consider
  \begin{align*}
    \mathcal{I} & := \set{f(Z',X') \mid X' \subseteq
      X, Z' \subseteq Z \text{ and } X',Z'
      \text{ are \ksmall}}\,.
  \end{align*}
  Since
   $Z_{n} \subseteq f(Z,X) 
     = \bigcup \mathcal{I}$  and  $\mathcal{I}$ is a \dirset, by
  Proposition~\ref{prop:ksmall} there exist $Z_{n+1},X_{n+1}$ \ksmall
  such that $Z_{n} \subseteq f(Z_{n+1},X_{n+1})$.
  Moreover, $Z_{n+1} \subseteq Z \subseteq f(Z,X)$.
    
  Let now $X_{\omega} := \bigcup_{n \geq 1}X_{n}$ and
  $Z_{\omega} := \bigcup_{n \geq 0}Z_{n}$. Notice that
  $Z_{\omega}$ and $X_{\omega}$ are \ksmall, since we assume
  that $\kappa > \aleph_{0}$.  We have therefore
  \begin{align*}
    Z_{\omega} &
    =  \bigcup_{n \geq 0}Z_{n}
    \subseteq  \bigcup_{n \geq 1}f(Z_{n},X_{n}) \subseteq f(\bigcup_{n \geq 1} Z_{n},\bigcup_{n \geq 1}
    X_{n}) 
    \subseteq f(Z_{\omega},X_{\omega})\,.
  \end{align*}
  Whence $b \in Z_{\omega} \subseteq \nu_{z}.f(z,X_{\omega})$,
  with $X_{\omega} \subseteq X$ and $X_{\omega}$
  \ksmall, proving that $\nu_{z}.f(z,-)$ is \continuous.
\end{proof}

\begin{prop}
  \label{prop:mucontinuous}
  Let $f : P(B)\times P(A) \rto P(B)$ be a \continuous monotone
  map. If $\kappa \geq \aleph_{0}$ then 
$\mu_{z}.f(z,\,_{-}) : P(A) \rto P(B)$ is also \continuous.
\end{prop}
\begin{proof}
  We suppose that $f$ is \continuous, 
  $\set{X_{i} \mid i \in I}$ is a \dirset of elements of $P(A)$
  and 
  $X = \bigcup_{i \in I} X_{i}$. We are going to show that
  $\mu_{x}.f(x,X) = \bigcup_{i \in I} \mu_{x}.f(x,X_{i})$.

  Firstly, notice that the relation
  $\mu_{x}.f(x,X) \supseteq \bigcup_{i \in I} \mu_{x}.f(x,X_{i})$
  follows from monotonicity; thus we only need to prove the converse
  relation and, to this end, it is enough to show that
  $\bigcup_{i \in I} \mu_{x}.f(x,X_{i})$ is a fixed-point of
  $f(x,X)$. This goes as follows:
  \begin{align*}
    f(\bigcup_{i \in I} \mu_{x}.f(x,X_{i}),X) & = \bigcup_{i \in I}
    f(\mu_{x}.f(x,X_{i}),X) \tag*{since
      $f$ is \continuous in its first argument}
    \\
    & = \bigcup_{i \in I} f(\mu_{x}.f(x,X_{i}), \bigcup_{j \in
      I}X_{j}) \\
    & = \bigcup_{i \in I,j \in I} f(\mu_{x}.f(x,X_{i}), X_{j})
    \tag*{since $f$ is \continuous in its second argument}
    \\
    & = \bigcup_{i \in I}  f(\mu_{x}.f(x,X_{i}), X_{i}) 
      \tag*{since $\set{X_{i} \mid i \in I}$ is \kdirected}
  \\
  & = \bigcup_{i \in I} \mu_{x}.f(x,X_{i})\,.
\end{align*}
\ENDOFPROOF[Proposition~\ref{prop:mucontinuous}]
\end{proof}
Notice that the statement of Proposition~\ref{prop:mucontinuous} holds
for \continuous 
monotone maps
$f : P\times Q \rto P$, that is, we might only assume that $P$ and $Q$
are complete lattices, not powerset
algebras. 
Indeed, the corresponding proof is obtained from the proof of
Proposition~\ref{prop:mucontinuous} by replacing 
the set theoretic $\bigcup$ with the supremum symbol
$\bigvee$. Similarly, the statement of
Proposition~\ref{prop:nucontinuous} is suitable to be generalized to
posets $P$ and $Q$ satisfying appropriate conditions, see
\cite{ITA2002}.

\medskip

Let ${\cal F} = \set{f_{i} : P(A)^{n_{i}} \rto P(A) \mid i \in I }$ be
a collection of 
monotone 
operations on $P(A)$. 
We 
define \emph{the $\mu$-clone of ${\cal F}$} to be the least set of
finitary operations on $P(A)$ that contains ${\cal F}$ and the
projections and which is closed under the following operations:
substitution, taking parametrized \LFP{s} and \GFP{s}.
\begin{cor}
  \label{cor:muclone}
  Let $\kappa > \aleph_{0}$ be a regular cardinal. If all the maps in
  ${\cal F}$ are \kcontinuous, then all the maps in the $\mu$-clone of
  ${\cal F}$ are also \kcontinuous.
\end{cor}
\begin{proof}
  We shall observe that projections are \kcontinuous and that the set
  of \kcontinuous functions is closed under substitution and under the
  operations of taking least 
  and \GFP{s}.  Projections are \Lower and \Uadjoint{s}, so they
  actually preserve all unions and intersections, see \cite[\S 7.23
  and Proposition~7.31]{DP}.  For substitution, argue first that the
  composition of two \kcontinuous maps is \kcontinuous. Observe then
  that if $f_{i} : P(A) \rto P(B_{i})$ is \kcontinuous, for
  $i = 1,\ldots ,n$, then the unique map
  $\langle f_{1},\ldots ,f_{n}\rangle : P(A) \rto \prod_{i}P(B_{i})$
  such that $\pi_{i} \circ \langle f_{1},\ldots ,f_{n}\rangle = f_{i}$
  for each $i = 1,\ldots ,n$, is \kcontinuous; this is because suprema
  are computed coordinatewise in $\prod_{i}P(B_{i})$.  Therefore, if
  $f_{0},f_{1},\ldots ,f_{n}$ are \kcontinuous, then also the
  \correction{substitution}{composite}
  $f_{0} \circ \langle f_{1},\ldots ,f_{n}\rangle$ is
  \kcontinuous. For least 
  and \GFP{s} use
  Propositions~\ref{prop:nucontinuous}~and~\ref{prop:mucontinuous}.
\end{proof}


\section{The propositional modal $\mu$-calculus}
\label{sec:mucalculus}
Here we present the propositional modal $\mu$-calculus and
some known results on this logic that we shall need later.

Hereinafter 
$\Act$ is a fixed finite set of actions and $Prop$ is a
countable set of propositional variables.
The set $\Lmu$ of formulas
of the propositional modal $\mu$-calculus over $\Act$ is generated by
the following grammar:
\begin{align}
  \label{grammar:mucalculus}
  \phi & \eqdef y \mid \neg y \mid \top \mid
  \phi \land \phi  \mid \bot \mid \phi \vee \phi \mid \pos[a] \phi
  \mid \nec[a] \phi
  \mid \mu_{z}.\phi \mid \nu_{z}.\phi\,,
\end{align}
where $a \in \Act$, $y \in Prop$, and $z \in Prop$ is a positive
variable in the formula $\phi$, \ie no occurrence of $z$ is under the
scope of a negation.  
  In general, we shall use $x,x_{1},\ldots ,x_{n},\ldots $ for
  variables that are never under the scope of a negation nor bound in
  a formula $\phi$; $y,y_{1},\ldots y_{n},\ldots $ for variables that
  are free in formulas; $z,z_{1},\ldots ,z_{n},\ldots $ for variables
  that are bound in formulas.  However, this convention cannot be
  rigorously enforced, since we shall often consider the steps from a
  formula $\phi$ with a free occurrence of the variable $z$ to the
  formula $\mu_{z}.\phi$, where $z$ is bound.
We think of the grammar~\eqref{grammar:mucalculus} as a way of
specifying the abstract syntax of a formula, as if it was the
specification of an inductive type in a programming language such as
Haskell. Nonetheless, we shall write formulas thus we need to be able
to disambiguate them. To achieve this goal we use standard
conventions: $\land$ has higher priority than $\vee$, unary modal
connectors
have higher priority than binary logical connectors. The least and
\GFP{s} operators yield priority instead, the dot notation emphasizes
this. For example, the formula $\mu_{x}.\phi \land \psi$ is implicitly
parenthesised as $\mu_{x}.(\phi \land \psi)$ instead of
$(\mu_{x}.\phi) \land \psi$.

  An $\Act$-\emph{model} (hereinafter referred to as model) is a
  triple $\M = \langle \UM,\Ras,v \rangle$ where: $\UM$ is a set (of
  worlds or states); for each $a \in \Act$,
  $R_{a} \subseteq \UM \times \UM$ is a (accessibility or transition)
  relation; $v : Prop\rto P(\UM)$ is a valuation, i.e., 
  an interpretation of the propositional variables as subsets of
  $\UM$.
Given a model ${\cal M}$, the semantics $\eval{\psi}$ of formulas
$\psi \in \Lmu$ as subsets of $\UM$ is recursively defined using the
standard clauses from \multimodal logic \K~(see e.g. \cite{gabbay2003many}).
For example, we have
\begin{align*}
  \eval{\pos[a]\psi}
  & \eqdef \set{s \in \UM \mid \exists s' \, (\,sR_{a}s' \ \, \&\ \,  s'\!\in \eval{\psi})}\,,
  \\  \eval{\nec[a]\psi}
& \eqdef \set{s \in \UM \mid \forall s' \,  (\,sR_{a}s' \,\Rightarrow \,  s' \in \eval{\psi}\,)}\,.
\end{align*}
\rephrase{Here we only define}{We present next} the semantics of the
least and greatest fixed-point constructors $\mu$ and $\nu$.  For this
purpose, given a subset $Z \subseteq \UM$, we define
$\variant{\pair{z}{Z}}$ to be the model that possibly differs from
$\M$ only on the value $Z$ that its valuation takes on $z$.  The
clauses for the fixed-point constructors are the following:
\begin{align*}
  \eval{\mu_{z}.\psi} & \eqdef \bigcap \set{Z \subseteq \UM \mid
    \eval[\variant{\pair{z}{Z}}]{\psi} \subseteq Z}\,,
  \\
  \eval{\nu_{z}.\psi} & \eqdef \bigcup \set{Z \subseteq \UM \mid Z
    \subseteq \eval[\variant{\pair{z}{Z}}]{\psi} }\,.
\end{align*}
A formula $\phi \in \Lmu$ and a variable $x \in Prop$ determine on
every model $\M$ the correspondence $\phiM^{x} : P(\UM) \rto P(\UM)$,
that sends each $S \subseteq \UM$ to
$\eval[\variant{\pair{x}{S}}]{\phi} \subseteq \UM$.  We shall write in
the following $\phiM$ for $\phiM^{x}$, when $x$ is understood.
Coming back to the clauses for the fixed-point constructors,
the syntactic restrictions on the variable $z$ in the productions of
$\mu_{z}.\psi$ and $\nu_{z}.\psi$ ($z$ must be positive in $\psi$)
imply that the function $\psiM^{z}$ is monotone.  By Tarski's theorem
\cite{tars:deci55}, the above clauses state that $\eval{\mu_{z}.\psi}$
and $\eval{\nu_{z}.\psi}$ are, respectively, the least and the
greatest fixed-point of $\psiM^{z}$.  As usual, we write
$\M,s \forces \psi$ to mean that $s \in \eval{\psi}$.

\subsection{The closure of a formula}
For $\phi \in \Lmu$, we denote by $Sub(\phi)$ the set of subformulas
of $\phi$.  A \emph{substitution} is an expression of the form
$[\psi_{1}/y_{1},\ldots ,\psi_{n}/y_{n}]$ where, for
$i =1\ldots ,n$, $y_{i}$ is a propositional variable and
  $\psi_{i} \in \Lmu$.  We use
  $\phi[\psi_{1}/y_{1},\ldots ,\psi_{n}/y_{n}]$ to denote
  \emph{application} of the substitution
  $[\psi_{1}/y_{1},\ldots ,\psi_{n}/y_{n}]$ to the formula
  $\phi$---that is, the result of simultaneously replacing every free
  occurrence of the variable $y_{i}$ in $\phi$ by the formula
  $\psi_{i}$, $i =1,\ldots ,n$.  As usual for formal systems with
  variable binders, we may assume that variable capture does not arise
  when applying substitutions to formulas.  When we want to emphasize
  application (of a substitution to a formula) we use a dot: for
  example, $\phi \cdot [\psi_{1}/y_{1},\ldots ,\psi_{n}/y_{n}]$ and
  $\phi[\psi_{1}/y_{1},\ldots ,\psi_{n}/y_{n}]$ denote the same
  formula.
  We also use the symbol $\cdot$ to denote composition of
  substitutions.
  For $\sigma_{1} \eqdef [\phi_{1}/x_{1},\ldots ,\phi_{n}/x_{n}]$ and
  $\sigma_{2} \eqdef [\psi_{1}/y_{1},\ldots ,\psi_{m}/y_{m}]$, the
  composite substitution $\sigma_{1} \cdot \sigma_{2}$ is defined by
\begin{align*}
  \sigma_{1} \cdot \sigma_{2} & \eqdef [\,\phi_{1} [\psi_{1}/y_{1},\ldots
  ,\psi_{m}/y_{m}]/x_{1}, 
  \ldots ,\phi_{n}[\psi_{1}/y_{1},\ldots
  ,\psi_{m}/y_{m}]/x_{n}\,]\,.
\end{align*}

\smallskip

A formula $\phi \in \Lmu$ is \emph{\wnamed} if no bound variable of
$\phi$ is also free in $\phi$ and, for each bound variable $z$ of
$\phi$, there is a unique subformula occurrence $\psi$ of $\phi$ of
the form $Q_{z}.\psi'$, with $Q \in \set{\mu,\nu}$.

It is well-known that every formula $\phi \in \Lmu$ is equivalent to a
\wnamed formula. 
We shall use \wnamed formulas only to have an accurate description of
the game  semantics, see \S~\ref{subsec:gameSemantics}.

\medskip

For $\phi \in \Lmu$ \wnamed and $\psi \in Sub(\phi)$, the
\emph{standard context} of $\psi$ in $\phi$ is
the 
composite substitution
\begin{align*}
 \sigmaphi_{\psi} & \eqdef   [Q^{n}_{z_{n}}.\psi_{n}/z_{n}] \cdot \myldots \cdot [Q^{1}_{z_{1}}.\psi_{1}/z_{1}]
\end{align*}
uniquely determined 
by the following conditions:
\begin{enumerate}
\item $\set{z_{1},\ldots ,z_{n}}$ is the set of variables that occur
  bound in $\phi$ and free in $\psi$,
\item for each $i = 1,\ldots ,n$, $Q^{i}_{z_{i}}.\psi_{i}$ is the
  unique subformula of $\phi$ such that $Q^{i} \in \set{\mu,\nu}$,
\item if $Q^{j}_{z_{j}}.\psi_{j}$ is a subformula of $\psi_{i}$, then
  $i < j$.
\end{enumerate}

The \emph{closure} of a \wnamed $\phi \in \Lmu$,
see \cite{Kozen83}, is the set $\CL(\phi)$ defined as follows:
\begin{align*}
  \CL(\phi) & \eqdef \set{ \psi \cdot \sigmaphi_{\psi} \mid \psi \in
    Sub(\phi)}\,.
\end{align*}
Recall from \cite{Kozen83} that $\CL(\phi)$ \rephrase{is}{can be
  characterised as} the least subset of $\Lmu$ such that
\begin{itemize}
\item $\phi \in \CL(\phi)$,
\item if $\psi_{1} @ \psi_{2} \in \CL(\phi)$, then
  $\psi_{1},\psi_{2} \in \CL(\phi)$, with $@ \in \set{\land,\vee}$,
\item if $\pos[a] \psi \in \CL(\phi)$ or $\nec[a] \psi \in \CL(\phi)$, then $\psi \in \CL(\phi)$,
\item if $Q_{z}.\psi \in \CL(\phi)$, then
  $\psi[Q_{z}.\psi/z] \in \CL(\phi)$, with $Q\in \set{\mu,\nu}$.
\end{itemize}
The definition of $\CL(\phi)$ implies it is finite.

\subsection{Game semantics}
\label{subsec:gameSemantics}
  
Given $\phi \in \Lmu$ \wnamed and a model
$\M = \langle \UM,\set{R_{a} \mid a \in \Act},v\rangle$, the game
$\G(\M,\phi)$ is the two player game of perfect information and
possibly infinite duration---a \emph{parity game}, see
e.g. \cite[Chapter 4]{AN01}---defined as follows.  Players of
$\G(\M,\phi)$ are named Eva and Adam.  The set of positions is the
Cartesian product $\UM \times \CL(\phi)$.
Moves are as in the table below:
$$
\begin{array}[t]{|c|c|}
  \hline
  \myrule\text{Adam's moves}
  &  \text{Eva's moves} \\
  \hline
  \begin{array}{r@{\hspace{1mm}}l}
    \myrule(s,\psi_{1} \land \psi_{2}) & \rto (s,\psi_{i})\,, \quad i = 1,2 \\
    \myrule(s,\nec[a] \psi) & \rto (s',\psi)\,, \quad s R_{a} s'\\
    \myrule(s,\nu_{z}.\psi) & \rto (s,\psi[\nu_{z}.\psi/z])
  \end{array}
  &
  \begin{array}{r@{\hspace{1mm}}l}
  \myrule(s,\psi_{1} \vee \psi_{2}) & \rto (s,\psi_{i}), \quad i = 1,2, \\
  \myrule(s,\pos[a] \psi) & \rto (s',\psi), \quad s R_{a} s', \\
  \myrule(s,\mu_{z}.\psi) & \rto (s,\psi[\mu_{z}.\psi/z])\,.
  \end{array} \\[30pt]  \hline
\end{array}
$$
From a position of the form $(s,\top)$ Adam loses, and from
a position of the form $(s,\bot)$ Eva loses.
Also, from a position of the form $(s,p)$ with $p$ a propositional
variable, Eva wins if and only if $s \in v(p)$; from a position of the
form $(s,\neg p)$ with $p$ a propositional variable, Eva wins if and
only if $s \not\in v(p)$.  The definition of the game is completed by
defining infinite winning plays. 
To achieve this goal, we choose a rank function
$\rho : \CL(\phi) \rto \N$ such that, when $\psi_{1}$ is a subformula
of $\psi_{2}$, then
$\rho(\psi_{1} \cdot \sigmaphi_{\psi_{1}}) \leq \rho(\psi_{2} \cdot
\sigmaphi_{\psi_{2}})$, and such that $\rho(\mu_{z}.\psi)$ is odd and
$\rho(\nu_{z}.\psi)$ is even. \correction{An infinite play}{The winner
  of an infinite play} $\set{(s_{n},\psi_{n}) \mid n \geq 0 }$ is
determined by the \emph{parity condition}: it is a win for Eva if and
only if
$\max \set{ n \geq 0 \mid \set{i \mid \rho^{-1}(\psi_{i}) \text{ is
      infinite} } } $ is even.

Let us recall the following fundamental
result (see for example~\cite[Theorem 6]{BraWal15}):
\begin{prop}
  For each model $\M$ and each \wnamed formula $\phi \in \Lmu$,
  $\M, s \forces \phi$ if and only if Eva has a winning strategy from
  position $(s,\phi)$ in the game $\G(\M,\phi)$.
\end{prop}

\subsection{Bisimulations}
Let $P \subseteq Prop$ be a subset of variables and let
$B \subseteq \Act$ be a subset of actions.  Let $\M$ and $\M'$ be two
models. A $(P,B)$-bisimulation is a relation
$\B \subseteq \UM \times \U{\M'}$ such that, for all $(x,x') \in \B$, we
have
\begin{itemize}
\item $x \in v(p)$ if and only if $x' \in v'(p)$, for all $p \in P$,
\item  for each $b \in B$,
  \begin{itemize}
  \item $xR_{b}y$ implies $x'R_{b}y'$ for some $y'$ such that
    $(y,y') \in \B$, 
  \item $x'R_{b}y'$ implies $xR_{b}y$ for some $y$ such that
    $(y,y') \in \B$.
  \end{itemize}

\end{itemize}
A pointed model is a pair $\tuple{\M,s}$ with
$\M = \langle \UM,\Ras, v\rangle$ a model and $s \in \UM$.  We say that
two pointed models $\tuple{\M,s}$ and $\tuple{\M',s'}$ are $(P,B)$-bisimilar if
there exists a $(P,B)$-bisimulation $\B \subseteq \UM \times \UM'$ with
$(s,s') \in \B$; we say that they are bisimilar if they are
$(Prop,\Act)$-bisimilar.

Let us denote by $\Lmu[P,B]$ the set of formulas whose free variables
are in $P$ and whose modalities are only indexed by actions in
$B$. The following statement is a straightforward refinement of
\cite[Theorem 10]{BraWal15}.
\begin{prop}
  \label{prop:bisimilar}
  If $\tuple{\M,s}$ and $\tuple{\M',s'}$ are $(P,B)$-bisimilar,
  then $\M,s \forces \phi$ if and only if $\M',s' \forces \phi$, for
  each $\phi \in \Lmu[P,B]$.
\end{prop}


\section{$\aleph_{1}$-continuous fragment of the modal $\mu$-calculus}
\label{sec:continuousfragment}

In this section we introduce a fragment of the modal $\mu$-calculus
which we name $\C(x)$.  Formulas in this fragment give rise to
\continuous[$\aleph_{1}$] maps when interpreted as monotone maps of
the variable $x$.  We show how to construct a formula
$\phi' \in \C(x)$ from a given arbitrary formula $\phi$ in order to
satisfy the following property: $\phi$ is \continuous for some
infinite regular cardinal $\kappa$
if and
only if $\phi$ and $\phi'$ are equivalent formulas. 
Our conclusions are twofold.  Firstly, we deduce the decidability of
the problem whether a formula is \continuous for some $\kappa$ is
decidable. Decidability relies on the effectiveness of the
construction and on the well-known fact that equivalence for the modal
$\mu$-calculus is elementary \cite{EmersonJutla91}. Secondly, we
observe that if a formula is \continuous, then it is already
\continuous[$\aleph_{1}$] or even \continuous[$\aleph_{0}$]. Thus,
there are no interesting notions of \kcontinuity for the modal
$\mu$-calculus besides those for the cardinals $\aleph_{0}$ and
$\aleph_{1}$.

\begin{defi}
  \label{def:kcontinuous}
  A formula \emph{$\phi\in \Lmu$ is \continuous in $x$} if $\phi_{\M}$
  is \continuous, for each model $\M$. If $X \subseteq Prop$, then we
  say that \emph{$\phi$ is \continuous in $X$} if $\phi$ is
  \continuous in $x$ for each $x \in X$.
\end{defi}

\begin{defi}
  We define $\C(X)$ to be the set of formulas of the modal
  $\mu$-calculus that can be generated by the following grammar:
  \begin{align}
    \label{grammar:C1}
    \phi & \eqdef x \mid \psi \mid \top \mid \bot \mid \phi \land \phi \
    \mid \phi \vee \phi \mid \pos[a] \phi \mid \mu_{z}.\chi \mid
    \nu_{z}.\chi\,,
  \end{align}
  where $x \in X$, $\psi \in \Lmu$ is a $\mu$-calculus formula not
  containing any variable $x \in X$, and
  $\chi \in \C(X \cup \set{z})$.
\end{defi}

If we omit the last production from the above grammar, we obtain a
grammar for the continuous fragment of the modal $\mu$-calculus, see
\cite{Fontaine08}, which we denote here by $\C[0](X)$.  For $i = 0,1$,
we shall write $\C[i](x)$ for $\C[i](\set{x})$. The main
\correction{achievement}{result} of \cite{Fontaine08} is 
that a formula $\phi \in \Lmu$ is \continuous[$\aleph_{0}$] in $x$ if
and only if it is equivalent to a formula in $\C[0](x)$.  It must be
observed that the fragment presented above is the same as the one
presented in \cite{FontaineVenema2018} under the name of finite width
fragment.

Let $X = \set{x_{1},\ldots ,x_{n}}$; a straightforward induction shows
that, for each $\phi \in \C(X)$, the map that sends a tuple
$(S_{1},\ldots ,S_{n}) \in P(\UM)^{n}$ to
$\eval[\variant{\pair{x_{1}}{S_{1}},\ldots
  ,\pair{x_{n}}{S_{n}}}]{\phi}$ belongs to the $\mu$-clone generated
by intersections, unions, the modal operators $\pos[a]_{\M}$ and the
constants $\eval{\psi}$. Since all these generating operations are
\continuous[$\aleph_{1}$] maps (actually, they are
\continuous[$\aleph_{0}$]) we can use Corollary~\ref{cor:muclone} to
derive the following statement.
\begin{prop}
  \label{prop:Conecontinuous}
  Every formula in the fragment $\C(X)$ is \continuous[$\aleph_{1}$]
  in $X$.
\end{prop}

\subsection{Syntactic considerations}
\begin{defi}
  The \emph{digraph $G(\phi)$} of a formula $\phi \in \Lmu$ is
  obtained from the syntax tree of $\phi$ by adding an edge from each
  occurrence of a bound variable to its binding fixed-point
  quantifier. The root of $G(\phi)$ is $\phi$.
\end{defi}

\begin{defi}
  A path in $G(\phi)$ is \emph{\bad} if one of its nodes corresponds to a subformula
  occurrence of the form $\nec[a]\psi$. A \bad cycle in $G(\phi)$ is a
  bad path starting and ending at the same vertex.
\end{defi}
Recall that a path in a digraph is simple if it does not visit twice
the same vertex. The rooted digraph $G(\phi)$ is a tree with
back-edges; in particular, it has the following property: for every node, there
exists a unique simple path from the root to this node.

\begin{defi}
  We say that an occurrence of a free variable $x$
  of $\phi$ is
  \begin{enumerate} 
  \item \emph{\bad} if there is a \bad path in $G(\phi)$ from the root
    to it;  
  \item \emph{\nsbad} (or \emph{\Boxed}) if the unique simple path
    in $G(\phi)$ from the root to it is \bad;
  \item \emph{\vbad} if it is \bad and not \Boxed.
  \end{enumerate}
\end{defi}

\begin{exa}
  Figure~\ref{fig:digraph} represents the digraph of the formula
  $$ (\mu_{z_1}. y_0 \land (\nu_{z_0}.z_0\land \nec {z_1}))\lor
  ( \pos
  {y_0} \land y_1)\,.
  $$
  From the figure we observe that:
  \begin{itemize}
  \item[{\tiny $\bullet$}] The free occurrence of $z_1$ in the digraph of
    $\nu_{z_0}.z_0\land {\nec {z_1}}$ (in dashed) is bad but
    \nsbad.
    
  \item[{\tiny $\bullet$}] The free occurrence of $y_0$ in the left branch of the
    digraph (in bold) is \vbad. The other occurrence of $y_0$ is not
    bad.
    
  \item[{\tiny $\bullet$}]
    The unique free occurrence of $y_1$ in $\phi$ is not bad.
\end{itemize}
\begin{figure}[h]
  \begin{center}
    \includegraphics[resolution=300]{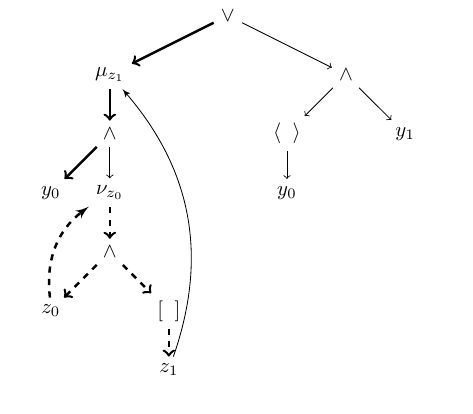}
    \caption{The digraph of a formula in $\Lmu$.}
    \label{fig:digraph}
  \end{center}
\end{figure}

\end{exa}

\begin{lem}
  \label{lemma:ConeSyntaxGraph}
For every set $X$ of variables and every $\phi\in \Lmu$, the following are equivalent:
  \begin{enumerate}
  \item  $\phi \in \C(X)$,
  \item no occurrence of a variable $x \in X$ is \bad in $\phi$.
  \end{enumerate}
\end{lem}

\begin{proof}
  Let $X$ be a set of variables and $\phi \in \Lmu$.
  
  (1) implies (2). The proof is by induction on the structure of
    formulas.
    Consider a formula $\phi \in \C(X)$ and observe that the only way
    to introduce a \bad path from the root of $G(\phi)$ to an
    occurrence of some variable $x \in X$ is either by using a modal
    operator $\nec[a]$---which, however, is excluded by the grammar
    defining the fragment $\C(X)$---or by a fixed-point formation
    rule.  Therefore, we focus on the case where $\phi$ is of the form
    $Q_{z}.\chi$, for $Q\in \{\mu,\nu\}$ and
    $\chi \in \C(X \cup \set{z})$, inductively assuming that no
    occurrence of a variable $x \in X \cup \set{z}$ is \bad in $\chi$.
    Suppose that there is an occurrence of a variable $x \in X$ and a
    \bad path from the root of $G(Q_{z}.\chi)$ to this occurrence.
    Since this occurrence of $x$ is not \bad in $\chi$, this path
    necessarily crosses an edge from an occurrence of the variable $z$
    to the root of $G(Q_{z}.\chi)$. But then this occurrence of $z$ is
    \bad in $G(\chi)$, contradicting the inductive
    hypothesis.

  (2) implies (1).  Suppose there exist pairs of the form
    $(X,\phi)$ where $X$ is a finite set of variables,
    $\phi \not\in \C(X)$ and, for each $x \in X$, $\phi$ has no \bad
    occurrence of $x$.  Among these pairs, consider $(X,\phi)$ with
    $\phi$ of least complexity, where we define the complexity of a
    formula $\phi$ as the number of vertices in $G(\phi)$. Clearly,
    $\phi$ has to be of the form $Q_{z}.\chi$.
    Moreover, by the second
    production of the grammar~\eqref{grammar:C1}, it must contain a
    free occurrence of a variable $x \in X$. 
    Observe that $\chi$ has no bad occurrence of any $x \in X$, since
    such a \bad occurrence yields a bad occurrence of $x$ in
    $Q_{z}.\chi$. Also, if an occurrence of $z$ is \bad in $\chi$,
    then any occurrence of some $x \in X$ is bad in
    $Q_{z}.\chi$. Therefore, $\chi$ has no \bad occurrence of any
    variable in $X \cup \set{z}$. By the minimality assumption on
    $(X,\phi)$, $\chi$ belongs to $\C(X \cup \set{z})$ and so
    $\phi \in \C(X)$, a contradiction.
\end{proof}

\subsection{The $\C(x)$-flattening of formulas}
We aim at defining the $\C(x)$-flattening 
 $\phi^{\flat x}$   of any formula
$\phi$ of the modal $\mu$-calculus.
This will go through the definition of the intermediate formula
$\lphi$ which has one more new free variable $\x$.  The formula
$\lphi$ is obtained from $\phi$ by renaming to $\x$ all the \Boxed
occurrences of the variable $x$.
In the definition of $\liftNp{\phi}$ below, we assume that $x$ has no
bound occurrences in $\phi$.
The formal definition is given by induction as follows:
\begin{align*}
  \liftNp{y} & = y &
  \lift{\neg y} & = \neg y \\
  \liftNp{\top} & = \top &\liftNp{\bot} & = \bot \\
  \lift{\psi_{0} @ \psi_{1}} & = \liftNp{\psi_{0}}  @\,
  \liftNp{\psi_{1}} \spacelesstext{\quad with $@\in \set{\land, \lor}$,} \\
  \lift{\pos[a]\psi} & = \pos[a] \liftNp{\psi}  & \lift{\nec[a]\psi} & 
  = \nec[a] \psi [\x/x] & &\\
   \lift{Q_{z}.\psi} & = Q_{z}.\liftNp{\psi} \spacelesstext{\qquad with $Q\in \set{\mu, \nu}$.}
\end{align*}
 The following fact is proved by a
straightforward induction.
\begin{lem}
  \label{lemma:phifromlift}
  For each $\phi\in \Lmu$, we have 
 \begin{align}
   \label{eq:phifromlift}
   \lphi \cdot \subst = \phi \,.
 \end{align}
\end{lem}

The $\C(x)$-flattening $\phi^{\flat x}$ of formula $\phi \in \Lmu$  is then
defined by:
\begin{align*}
 \phi^{\flat x}& \eqdef \lphi \cdot [\bot/\x]\,
\end{align*}
and henceforward we shorten it to $\fphi$.

Let us notice that $\lphi$ (or $\fphi$) does not in general belong to
$\C(x)$. For example,
$ \flatten{\mu_{z}.x \vee \nec[a] z} = \mu_{z}.x \vee \nec[a] z
\not\in \C(x)$ since $x \vee \nec[a] z \not\in \C(\{x,z\})$. Yet,
the following definition and lemma partially justify the choice of
naming.
\begin{defi}
  A formula $\phi$ is \emph{\good} w.r.t. a set $X$ of variables if no
  occurrence of a variable $x \in X$ is \vbad. A formula $\phi$ is
  \emph{\good} if it is \good w.r.t. $\set{x}$.
\end{defi}

\begin{rem}
  \label{rem:presagood}
  Let $\psi$ be a \wnamed variant of a formula $\phi$, so $\psi$ is
  obtained from $\phi$ by renaming some bound variables. The digraphs
  $G(\psi)$ and $G(\phi)$ differ only for the labelling of some pairs
  of nodes lying on a back edge from an occurrence of a bound variable
  to its binding fixed-point quantifier.  Now let $P$ be a property of
  formulas defined by means of the digraphs $G(\phi)$ without
  mentioning the labels of nodes on any of those back-edges. Then a
  formula $\phi$ has the property $P$ if and only if any of its
  \wnamed variant has the property $P$.  One such $P$ is the property
  of being \agood. Therefore, if $\phi$ is \agood, then so it is any
  of its \wnamed variants.
\end{rem}

\begin{lem}
  \label{lemma:almostGoodOK}
  If $\phi$ is an \good formula, then both $\lphi$ and $\fphi$ belong
  to $\C(x)$.
\end{lem}
\begin{proof}
  We prove the result for $\lphi$.  Consider a \bad occurrence of $x$
  in $\lphi$. After substituting $\x$ for $x$, such an occurrence
  yields a bad occurrence of $x$ in $\phi$. Since there are no \vbad
  occurrences of $x$ in $\phi$, then this occurrence should be \nsbad,
  that is, under the scope of a necessity modal operator $\nec[a]$.
  But then this same occurrence of $x$ in $\phi$ would correspond to
  an occurrence of $\x$ in $\lphi$ and not to an occurrence of $x$ as
  assumed.
\end{proof}

We aim 
to transform a formula $\phi$ into an equivalent formula in which
there are no \vbad occurrences of the variable $x$. The transformation
that we define next achieves this goal.  For $\phi \in \Lmu$ and a
finite set $X$ of variables not bound in $\phi$, we define a formula
$\boxingNp[X]{\psi}$, with all the occurrences of a bad variable
$x \in X$ \Boxed (aka \nsbad).  
We let
\begin{align*}
  \boxingNp[X]{\psi} & \eqdef \psi\,,
  \quad\text{if no occurrence of a variable $x \in X$ is
    \vbad in $\psi$,}
  \intertext{and, otherwise,}
  \boxing[X]{\pos[a] \psi} & \eqdef \pos[a] \boxing[X]{\psi}\,, \\
  \boxing[X]{\psi_{1} \loperator \psi_{2}} & \eqdef \boxing[X]{\psi_{1}}
  \loperator \boxing[X]{\psi_{2}}\,,
  \quad\text{with $ \loperator \in \set{\land,\lor}$,} \\
  \boxing[X]{Q_{z}.\psi} & \eqdef \psi_{0}[\psi_{1}/\boxed{z}]\,,
  \spacelesstext{\quad where}
  \\
  & \psi_{0} \eqdef Q_{z}.\psi_{2},\; \psi_{2} \eqdef \lift[z]{\boxingNp[X \cup
    \sset{z}]{\psi}}\,,
    \tand
    \psi_{1} \eqdef
    Q_{\boxed{z}}.\psi_{0}\,,
\end{align*}
with $Q \in \set{\mu,\nu}$.
That is, in
the last clause, $\psi_{2}$ is obtained from
$\boxingNp[X \cup \sset{z}]{\psi}$ by renaming all the \Boxed
occurrences of $z$ to $\boxed{z}$.
  A key point of the definition of $\boxing[X]{Q_{z}.\psi}$ is that,
  when we split, with $\psi_{2}$, the fixed-point variable $z$ into
  its \Boxed/un\Boxed parts, we also split, with $\psi_{1}$ and
  $\psi_{0}$, the respective fixed-point bindings, see
  Figure~\ref{fig:proofOfLemma25}. 
Observe that the first defining clause implies that
\begin{align*}
    \boxingNp[X]{x} & = x\ \text{if $x \in X$,}\\
    \boxingNp[X]{\psi} & = \psi\ \text{if $\psi$ contains no
      variable
      $x \in X$,} 
    \\
    \boxing[X]{\nec[a]\psi} & = \nec[a] \psi\,.
  \end{align*}
  \begin{exa}
    Consider the formula
    $\psi \eqdef x\lor \mu_{z}.\,x\lor z\lor \nec[a] (x\land z)$,
    where only the second occurrence of $x$ is \vbad.  For
    $X=\set{x}$ we have
    \begin{align*}
      \boxingNp[X]{\psi}= x\lor \mu_{z}.x\lor z\lor \nec[a] (x\land
      \mu_{\boxed{z}}.\mu_{z}.x\lor z\lor \nec[a] (x\land
      \boxed{z}))
    \end{align*}
    where no occurrence of $x$ is \vbad and so the formula
    $\boxingNp[X]{\psi}$ is \agood.
     \end{exa}

  \begin{prop}
    \label{prop:boxphieqphi}
    \label{prop:boxphialmostgood}
    The formula $\boxingNp[X]{\phi}$ is \good w.r.t. $X$ and it is
    equivalent to the formula $\phi$.
  \end{prop}
  We split the proof of the proposition in two lemmas.
  \begin{lem}
    The formula $\boxingNp[X]{\phi}$ is
    equivalent to $\phi$.
  \end{lem}
  \begin{proof}
    The statement of the proposition is obvious if a formula matches
    the base case of the definition. Also, in the cases of a modal
    formula $\pos[a]\psi$ and of a formula
    $\psi_{1} \loperator \psi_{2}$ with
    $\loperator \in \set{\land,\lor}$, the statement is an immediate
    consequence of the inductive hypothesis.
    In case of a formula of the form $\boxing[X]{Q_{z}.\psi}$ with
    $Q \in \set{\mu,\nu}$, we argue as follows:
    \begin{align*}
      \boxing[X]{Q_{z}.\psi}
      =  \psi_{0}[Q_{\boxed{z}}.\psi_{0}/\boxed{z}] 
      &\equiv Q_{\boxed{z}}.\psi_{0}\,, \tag*{by the fixed-point
        equation,}  \\
      = Q_{\boxed{z}}.Q_{z}.\psi_{2} & \equiv
      Q_{z}.\psi_{2}[z/\boxed{z}]\,,\tag*{by the equational properties
        of
        fixed-points,}  \\
      = Q_{z}.(\lift[z]{\boxingNp[X \cup
        \sset{z}]{\psi}}[z/\boxed{z}]) & = Q_{z}.(\boxingNp[X \cup
      \sset{z}]{\psi})\,, \tag*{by equation~\eqref{eq:phifromlift},}
      \\
      & \equiv Q_{z}.\psi\,, \tag*{by the inductive hypothesis.}
    \end{align*}
  \end{proof}

  \begin{lem}
    \label{lemma:boxedX}
    The formula $\boxingNp[X]{\phi}$ is \good, that is, it has no \vbad
    occurrence of a variable $x \in X$.
  \end{lem}
  Figure~\ref{fig:proofOfLemma25} illustrates the proof of
  this lemma.
  \begin{proof}
    The statement of the proposition is obvious if a formula matches
    the base case of the definition. Also, in the cases of a modal
    formula $\pos[a]\psi$ and of a formula
    $\psi_{1} \loperator \psi_{2}$ with
    $\loperator \in \set{\land,\lor}$, the statement is an immediate
    consequence of the inductive hypothesis.
    The only non-trivial case is that of a formula of the form
    $\boxing[X]{Q_{z}.\psi}$ with $Q \in \set{\mu,\nu}$.

    Let us firstly recall that $\boxing[X]{Q_{z}.\psi}$ is of the form
    $\psi_{0}[Q_{\boxed{z}}.\psi_{0}/\boxed{z}]$ with
    $\psi_{0} = Q_{z}.\psi_{2}$ and
    $\psi_{2} = \lift[z]{\boxingNp[X \cup \sset{z}]{\psi}}$.  Also,
    for the sake of readability, we have let
    $\psi_{1} \eqdef Q_{\boxed{z}}.\psi_{0}$ in the definition, so
    $\boxing[X]{Q_{z}.\psi} = \psi_{0}[\psi_{1}/\boxed{z}]$.  In
    particular, every occurrence of a variable $x \in X$ is located
    within $\psi_{0}$, or it is located in some subtree of
    $\psi_{0}[\psi_{1}/\boxed{z}]$ rooted at some occurrence of the
    subformula $\psi_{1}$.

    We argue next that every occurrence of a variable $x \in X$ within
    $\psi_{0} = Q_{z}.\psi_{2}$ is not \vbad.  By the induction
    hypothesis, such an occurrence of $x$ is not \vbad within
    $\psi_{2}$; the only reason for becoming \vbad in $\psi_{0}$ is
    then the existence of a cycle going through an edge from some
    occurrence of the variable $z$ to the formula
    $Q_{z}.\psi_{2}$. Such a bad cycle can arise for two reasons:
    either (a) there is a necessity modal operator $\nec[a]$ from
    $\psi_{2}$ to this occurrence of $z$, or (b) there is a bad cycle
    in some subformula of $\psi_{2}$ of the form $Q_{w}.\chi$, with
    this subformula lying on the path from $\psi_{2}$ to the
    occurrence of $z$.
    Yet (a) is not possible: recall that
    $\psi_{2} = \lift[z]{\boxingNp[X \cup \sset{z}]{\psi}}$, thus all
    the occurrences of $z$ within $\psi_{2}$ are not \Boxed (such
    an occurrence in $\boxingNp[X \cup \sset{z}]{\psi}$ has been
    renamed to $\boxed{z}$ in $\psi_{2}$).
    Also (b) is not possible, since otherwise the occurrence of $z$ in
    $\psi_{2}$ is \vbad. Yet we know that the same occurrence of $z$
    is not \vbad in $\boxingNp[X \cup \set{z}]{\psi}$, and renaming
    the \Boxed occurrences of $z$ to $\boxed{z}$ in this formula
    cannot transform another occurrence of $z$ into a \vbad occurrence. 
    
    Finally, we argue that there is no \vbad occurrence of some
    variable $x \in X$ in $\psi_{0}[\psi_{1}/\boxed{z}]$.  Suppose
    there is such an occurrence of $x$. If this occurrence is located
    within $\psi_{0}$, then this would also be a bad occurrence for
    $\psi_{0}$, which we have excluded. Thus, such an occurrence is
    located within some occurrence of the subformula $\psi_{1}$. But
    since every occurrence of the variable $\boxed{z}$ within
    $\psi_{0}$ is \Boxed, all the variable occurrences of $x$ within
    $\psi_{1}$ become \Boxed in the formula
    $\psi_{0}[\psi_{1}/\boxed{z}]$.

    Therefore, no occurrence of $x \in X$ is \vbad in $\psi_{0}[\psi_{1}/\boxed{z}]$.
\end{proof}

\newcommand{\psiZ}[1]{
  \node at (0,4) {$Q.z$} ;
  \draw[dashed] (0,3.6) -- (0,3) ;
  \draw (0,3) -- (-1.5,0) -- (1.5,0) -- (0,3);
  \node[fill=white,inner sep=0.5pt] (x) at (-1.1,0) {$x$} ;
  \node[fill=white,inner sep=0.5pt] (zbar) at (-0.2,0) {#1} ;
  \node[fill=white,inner sep=0.5pt] (z) at (1.0,0) {$z$} ;
  \draw [->, dashed] (z)++(0,-0.2) .. controls (2,-1) and (2,1).. (0.4,3.8);
  \draw[dashed] (0,3) -- (0,2) -- (z);
  \draw[dashed] (0,2) -- (zbar);
  \draw[dashed] (0,2.5) -- (x);
  \node[fill=white,inner sep=0.5pt,scale=0.5] at (-0.1,1) {$\nec[a]$};
  \draw[decorate, decoration={mirror}, thick]  (z)++(1,-0.2) -- ++(0,3.3);
  \node at (2.8,1.5) {$\psi_{2}$};
}
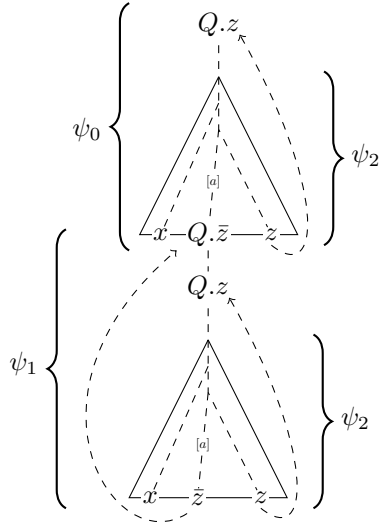
\begin{figure}[h]
  \centering
  \begin{tikzpicture}[xscale=0.7,yscale=.7, decoration={brace, amplitude=7pt}]
   \clip(-4.0,-5.5) rectangle (3.1,4.7);
    \psiZ{$Q.\bar{z}$}
    \coordinate (Qz) at (zbar) ;
    \draw[dashed] (Qz)+(0,-0.3) -- +(0,-0.7) ;
    \coordinate (newbase) at ++(-0.2,-5);
    \begin{scope}[shift={(newbase)}]
      \psiZ{$\bar{z}$}
    \end{scope}
    (Qz) \coordinate (goal) at ++(-0.8,-0.3);
    \draw [->, dashed] (zbar)+(-0.2,-0.3) .. controls (-2,-6.2)  and (-4,-3)
    .. (goal) ;
    \draw[thick, decorate] (newbase)++(-2.7,-0.2) -- ++(0,5.3) ;
    \node at (-3.7,-2.45) {$\psi_{1}$};
    \draw[thick, decorate] (-1.7,-0.3) -- (-1.7,4.4);
    \node at (-2.5,2) {$\psi_{0}$};
  \end{tikzpicture}
  \caption{Illustration of the proof of Lemma~\ref{lemma:boxedX}}
  \label{fig:proofOfLemma25}
\end{figure}

We can finally state our first main result.
\begin{thm}
  \label{cor:flatinC}
  \label{thm:flatinC}
  Every formula $\phi$ is equivalent to a formula $\psi$ with
  $\liftNp{\psi}$ and $\flattenNp{\psi}$ in $\C(x)$.  Moreover,
    we can choose $\psi$ \wnamed.
\end{thm}
In the theorem we can take $\psi$ to be a \wnamed variant of the
\agood formula $\boxingNp[\set{x}]{\phi}$. Then, by
Remark~\ref{rem:presagood}, $\psi$ is \agood and therefore, by
Lemma~\ref{lemma:almostGoodOK}, $\liftNp{\psi}$ and $\flattenNp{\psi}$
belong to $\C(x)$.

\subsection{Comparing the closures of $\phi$ and $\fphi$}

In the following, let $\phi$ be a \wnamed formula. Observe that both
$\liftNp{\phi}$ and $\flattenNp{\phi}$ are also \wnamed---we verify this
for $\liftNp{\phi}$, the argument for $\flattenNp{\phi}$ is similar.
Indeed, $\liftNp{\phi}$ has the same bound variables as $\phi$ and
therefore $\x$, assumed to be new, cannot be bound in
$\liftNp{\phi}$. Next, if $z$ is bound in $\liftNp{\phi}$, then it is
bound in $\phi$ and there is a unique subformula occurrence of $\phi$
of the form $Q_{z}.\psi$ and therefore a unique subformula occurrence
of $\liftNp{\phi}$ of the form $Q_{z}.\psi'$, the latter being either
$Q_{z}.\liftNp{\psi}$ or $Q_{z}.\psi[\x/x]$.

We develop here some syntactic considerations that allow
us to relate the closures of $\phi$ and $\fphi$. In turn, that will
make it possible to relate the positions of the games $\G(\M,\phi)$
and $\G(\M,\fphi)$, and so to construct, in the proof of
Proposition~\ref{prop:maincontinuous}, a winning strategy in the
latter game from a winning strategy in the former.

Recall that 
we use $Sub(\phi)$ for the set of subformulas of $\phi$.
\begin{rem}
  If $\x$ and $y$ are distinct variables and 
  $\chi$ is a formula that does not contain the variable $y$, then
  \begin{align}
    \label{eq:halfcomm}
    [\psi/y] \cdot [\chi/\x] & = [\chi/\x]\cdot [\psi[\chi/\x]/y]\,.
  \end{align}
  Also, if $\x$ is a variable occurring free in $\phi$ and $\Const$ is
  either a variable or a constant, then
  $Sub(\phi\cdot[\Const/\x])  = \set{\psi\cdot[\Const/\x]\mid \psi
      \in Sub(\phi)}$.
\end{rem}
The above remark is easily justified considering that for terms $t,s$
over an arbitrary signature we have
$Sub(t[s/\x]) = \set{t'[s/\x] \mid t' \in Sub(t)} \cup Sub(s)$,
whenever $\x$ is a variable occurring free in $t$, where now $Sub(t)$
denotes the set of subterms of $t$.

\begin{lem}
  \label{lemma:closuresubst}
  If $\x$ is a free variable of $\phi$ and $\Const$ is either a
  variable not bound in $\phi$ or a constant, then
  \begin{align*}
    \CL(\phi\cdot[\Const/\x]) & = \set{\psi\cdot[\Const/\x]\mid \psi \in
      \CL(\phi)}\,.
  \intertext{In particular, we have}
    \CL(\phi) & = \set{\phi'\cdot \subst \mid \phi' \in \CL(\lphi)
    }\,, 
    &
    \CL(\fphi) & = \set{\phi'\cdot [\bot/\x] \mid \phi' \in
      \CL(\lphi) }\,.
  \end{align*}
\end{lem}
The second statement of the lemma is an immediate consequence of the
first, considering that $\phi = \lphi \cdot \subst$ and
$\fphi = \lphi\cdot [\bot/\x]$.
\begin{proof}
  By repeatedly using equation~\eqref{eq:halfcomm} with
  $\chi = \Const$, we have
  \begin{align*}
    {\sigma^{\phi}_{\psi} \cdot\substp} 
    & = [Q_{n}y_{n}.\psi_{n}/y_{n}]\cdot \myldots
  \cdot
    [Q_{1}y_{1}.\psi_{1}/y_{1}] \cdot\substp \\
    & =
    \substp \cdot [Q_{n}y_{n}.\psi_{n}\cdot\substp/y_{n}] \cdot \myldots
    \cdot
    [Q_{1}y_{1}.\psi_{1}\cdot\substp/y_{1}] \,.
  \end{align*}
  Inspection of the three properties defining the standard context
  $\sigmaphi_{\psi}$ shows that the equality
  \begin{align*}
    \sigmaphis_{\psi\cdot\substp} & =
    [Q_{n}y_{n}.\psi_{n}\cdot\substp/y_{n}]
    \cdot \myldots
    \cdot
    [Q_{1}y_{1}.\psi_{1}\cdot\substp/y_{1}] 
  \end{align*}
  holds.  From this we deduce
  \begin{align}
    \label{eq:subststandardcontext}
    (\psi \cdot \substp) \cdot  \sigmaphis_{\psi\cdot\substp}
    & = (\psi \cdot \sigmaphi_{\psi}) \cdot \substp\,.
  \end{align}
  \begin{align*}
    \mbox{Thus }{\phi' \in \CL(\phi \cdot \substp)}  &\tiff \phi' = \psi \cdot
    \sigmaphis_{\psi} \text{ for some $\psi
      \in Sub(\phi \cdot \substp)$} \\
    & \tiff \phi' = \psi \cdot
    \substp \cdot \sigmaphis_{\psi \cdot \substp} \text{ for some
      $\psi
      \in Sub(\phi)$} \\
    & \tiff \phi' = \psi  \cdot \sigmaphi_{\psi}\cdot
    \substp  \text{ for some
      $\psi
      \in Sub(\phi)$} \\
    & \tiff \phi' = \phi''\cdot
    \substp  \text{ for some
      $\phi''
      \in \CL(\phi)$.}
  \end{align*}
  \ENDOFPROOF[Lemma~\ref{lemma:closuresubst}]
\end{proof}

\subsection{The continuous fragments}
Now we aim to prove
some sort of converse of Proposition~\ref{prop:Conecontinuous}, namely
that every \continuous formula $\phi$ of the propositional modal
$\mu$-calculus is equivalent to $\fphi$, where $\kappa$ is still
assumed to be an infinite regular cardinal.

A pointed model $\langle\M,s\rangle$ is a \emph{tree model} if the
rooted digraph $\langle \UM,\bigcup_{a \in \Act} R_{a},s\rangle$ is a
tree.  Let $\kappa$ be a cardinal. A tree model $\tuple{\M,s}$ is
\emph{\expanded} if, for each $a \in \Act$, whenever $x R_{a} x'$,
there are at least $\kappa$ $a$-successors of $x$ that are bisimilar
to $x'$. The following lemma is straightforward, see e.g.
\cite[Proposition 1]{Fontaine08} for the case where
$\kappa = \aleph_{0}$.
\begin{lem}
  \label{lemma:expanded}
  For each pointed model $\tuple{\M,s}$ there exists a \expanded tree
  model $\tuple{\T,t}$ bisimilar to $\tuple{\M,s}$.
\end{lem}

\begin{prop}
  \label{prop:maincontinuous}
  If $\M,s \forces \phi$ and $\phi$ is \continuous in $x$, then
  $\M,s \forces \fphi$.
\end{prop}
\begin{proof}
  Suppose that $\M= (\UM,\set{R_{a} \mid a \in A},v)$ is a model and
  that $s_{0} \forces \phi$. We want to prove that
  $s_{0} \forces \fphi$.  Notice first that, by
  Lemma~\ref{lemma:expanded}, we can assume that $\tuple{\M,s_{0}}$ is
  a \expanded tree model.

  Since $\phi$ is \continuous in $x$ and $s_0\in\phi_{\M}(v(x))$,
  there exists $U\subseteq v(x)$, with cardinality of $U$ strictly
  smaller than $\kappa$, such that $s_0\in\phi_{\M}(U)$, so
  $\M[\pair{x}{U}],s_{0} \forces \phi$.  We shall argue that
  $\M[\pair{x}{U}],s_{0} \forces \fphi$, from which it follows that
  $s_{0} \in \fphi_{\M}(U) \subseteq \fphi_{\M}(v(x))$---since
  $\fphi_{\M}$ is monotone---thus $\M,s_{0} \forces \fphi$.

  In the following let $\NN = \M[\pair{x}{U}]$ (notice that $\NN$ is not
  anymore \expanded).
  Since $\NN,s_{0} \forces \phi$, let us fix a winning strategy for
  Eva in the game $\G(\NN,\phi)$ from position $(s_{0},\phi)$. We
  define next a strategy for Eva in the game $\G(\NN,\fphi)$ from
  position $(s_{0},\fphi)$.  Observe first that, by
  Lemma~\ref{lemma:closuresubst}, positions in $\G(\NN,\phi)$
  (respectively, $\G(\NN,\fphi)$) are of the form $(s,\psi[x/\x])$
  (resp., $(s,\psi[\bot/\x])$) for a formula $\psi \in \CL(\lphi)$.
  Therefore, at the beginning of the play, Eva plays in
  $\G(\NN,\fphi)$ simulating the moves of the given winning strategy
  for the game $\G(\NN,\phi)$. The simulation goes on until the play
  reaches a pair of positions
  $p \eqdef (s,\nec[a]\chi \sigma^{\lphi}_{\nec[a]\chi}\cdot [x/\x])$
  and
  $p' \eqdef (s,\nec[a]\chi \sigma^{\lphi}_{\nec[a]\chi}\cdot
  [\bot/\x])$, for some subformula $\nec[a]\chi$ of $\lphi$, where
  $\chi = \chi'[\x/x]$ for some subformula $\chi'$ of $\phi$.
  \begin{clm}
    The positions $p$ and $p'$ are respectively of the form
    $(s,\nec[a] \psi) \in \G(\NN,\phi)$ and
    $(s,\nec[a] \psi') \in \G(\NN,\fphi)$ for some $\psi$ and $\psi'$
    such that $\psi[\bot/x] \impl \psi'$ is a tautology.
  \end{clm}
  \begin{proofofclaimNoQed}
    In the computations that follows we use the notation
    $\phi \geq \phi'$ (for $\phi,\phi'\in\Lmu$) to mean that
    $\eval{\phi} \supseteq \eval{\phi'}$ for every $\M$ (i.e.,
    $\phi' \impl \phi$ is a tautology).
    
    We let $\psi \eqdef \chi \sigma^{\lphi}_{\chi}\cdot [x/\x]$ and
    observe that
    \begin{align*}
      \psi & = \chi \sigma^{\lphi}_{\chi}\cdot [x/\x] 
      = \chi' [\x/x] \cdot \sigma^{\lphi}_{\chi' [\x/x]}\cdot [x/\x]\\
    & = \chi' [\x/x] \cdot [x/\x] \cdot \sigma^{\lphi \cdot
      [x/\x]}_{\chi' [\x/x]\cdot [x/\x]}\,, \tag*{by
      equation~\eqref{eq:subststandardcontext},}
    \\
    & = \chi' \cdot \sigmaphi_{\chi'} \,,
  \end{align*}

  On the other hand, we let $\psi' \eqdef \chi \sigma^{\lphi}_{\chi}\cdot
  [\bot/\x]$, so that
  \begin{align*}
    \psi' & = \chi \sigma^{\lphi}_{\chi}\cdot [\bot/\x] 
    = \chi' [\x/x] \cdot \sigma^{\lphi}_{\chi' [\x/x]}\cdot
    [\bot/\x] 
    = \chi' [\x/x] \cdot [\bot/\x] \cdot \sigma^{\lphi \cdot
      [\bot/\x]}_{\chi' [\x/x]\cdot [\bot/\x]}
    \\
    & = \chi' [\bot/x] \cdot \sigma^{\lphi \cdot [\bot/\x]}_{\chi'
      [\bot/x]}\,, \tag*{since
      $\chi'$ does not contain the variable $\x$,}
    \\
    & \geq \chi' [\bot/x] \cdot \sigma^{\lphi \cdot
      [\bot/\x,\bot/x]}_{\chi' [\bot/x]} \,,\tag*{since
      $[\bot/\x] \geq [\bot/\x,\bot/x]$ and
      $\lphi$ is monotone in $x$ and $\x$,}
    \\
 & = \chi'
    [\bot/x] \cdot \sigma^{\lphi \cdot [x/\x]\cdot[\bot/x]}_{\chi'
      [\bot/x]} = \chi' [\bot/x] \cdot \sigma^{\phi [\bot/x]}_{\chi'
      [\bot/x]} 
    \\
    & = \chi' \cdot \sigmaphi_{\chi'} \cdot [\bot/x] = \psi
    [\bot/x]\,, \tag*{by the previous computations. \qedClaim}
  \end{align*}
\end{proofofclaimNoQed}

Thus, Eva needs to continue playing in the
game $\G(\NN,\fphi)$ from
  a position
  of the form $(s,\nec[a] \psi')$ where $\psi[\bot/x] \impl \psi'$ is
  a tautology.
  We construct a winning stategy for Eva from this position as
  follows. Since the play has reached the position $(s,\nec[a]\psi)$
  of $\G(\NN,\phi)$ we also know that $s\in \evalN{\nec[a]\psi}$. We
  argue then that $s\in \evalN{\nec[a]\psi}$ implies
  $s\in \evalN{\nec[a]\psi[\bot/x]}$. Since
  $\evalN{\nec[a]\psi[\bot/x]}\subseteq \evalN{\nec[a]\psi'}$, Eva
  also has a winning strategy from position $(s,\nec[a]\psi')$ of the
  game $\G(\NN,\fphi)$, which she shall use to continue the play.

  \begin{clm}
    $s\in \evalN{\nec[a]\psi}$ implies
    $s\in \evalN{\nec[a]\psi[\bot/x]}$.
  \end{clm}
  \begin{proofofclaim}
    The statement of the claim trivially holds if $s$ has no
    successors.  Let $s'$ be a fixed $a$-successor of $s$ (i.e.
    $s R_{a} s'$), so $\NN,s' \forces \psi$; we want to show that
    $\NN,s' \forces \psi[\bot/x]$.
    To this goal, recalling that
    $\psi[\bot/x]\in \Lmu[Prop\setminus\{x\},\Act]$ and using
    Proposition~\ref{prop:bisimilar}, it is enough to prove that
    $\tuple{\NN,s'}$ is $(Prop\setminus\{x\},\Act)$-bisimilar to some
    $\tuple{\NN,s''}$ such that $\NN,s'' \forces \psi[\bot/x]$.
    
    Let  $S$ be the set    
    \begin{align*}
      & \set{t \mid s R_{a} t, \text{ $\tuple{\M,t}$ is bisimilar to
          $\tuple{\M,s'}$}, \tand \downset{t} \cap U \neq \emptyset},
    \end{align*}
    where we have used $\downset{t}$ to denote
    the subtree of $\tuple{\M,s_{0}}$ rooted at $t$.  Recall that the
    cardinality of $U$ is strictly smaller than $\kappa$ and so is the
    cardinality of $S$ once it is at most equal to the cardinality of
    $U$. But the cardinality of
    $\set{ t\mid s R_{a} t, \text{ $\tuple{\M,t}$ is bisimilar to
        $\tuple{\M,s'}$}}$ is at least $\kappa$ (recall
    $\tuple{\M,s_{0}}$ is a \expanded tree model). Consequently, there
    must be a successor $s''$ of $s$ such that $\tuple{\M,s''}$ is
    bisimilar to $\tuple{\M,s'}$ and which does not belong to $S$,
    that is $\downset s'' \cap U = \emptyset$ (i.e. no states in $U$
    are reachable from $s''$).  
    Since $\NN,s'' \forces \psi$ and
    $\downset s'' \cap U = \emptyset$, we have
    $\NN,s'' \forces \psi[\bot/x]$.  Yet $\tuple{\M,s''}$ and
    $\tuple{\M,s'}$ are bisimilar and since $\NN$ is obtained from
    $\M$ just by modifying the value of the variable $x$,
    $\tuple{\NN,s''}$ and $\tuple{\NN,s'}$ are
    $(Prop\setminus\set{x},\Act)$-bisimilar.  As stated before, this
    and $\NN,s'' \forces \psi[\bot/x]$ imply 
    $\NN,s' \forces \psi[\bot/x]$.
  \end{proofofclaim}
  To complete the proof of Proposition~\ref{prop:maincontinuous} we
  need to argue that the strategy so defined for Eva to play in the
  game $\G(\M,\fphi)$ is winning. The only difficulty in asserting
  this is to exclude the case where the initial simulation leads to a pair
  of positions of the form $(s,\x[x/\x])$ and $(s,\x[\bot/\x])$. This
  is however excluded since in $\lphi$ all the occurrences of $\x$ are
  \Boxed, so we are enforced to go through the second step of the
  strategy.
\end{proof}

\begin{prop}
  \label{prop:maincontinuoustwo}
  If, for some regular cardinal $\kappa$, $\phi \in \Lmu$ is \continuous, then
  $\phi$ is equivalent to $\fphi$.
\end{prop}
\begin{proof}
  Notice that, by monotonicity in the variable $x$, $\fphi \impl \phi$
  is a tautology. Proposition~\ref{prop:maincontinuous} exhibits the
  converse implication as another tautology.
\end{proof}

\begin{thm}
  \label{thm:mainone}
  If for some regular cardinal $\kappa$, $\phi \in \Lmu$ is a
  \continuous formula, then $\phi$ is equivalent to a formula
  $\phi' \in\C(x)$.
\end{thm}
\begin{proof}
  Suppose that $\phi$ is \continuous. By Corollary~\ref{cor:flatinC},
  $\phi$ is equivalent to a formula $\psi$ with
  $\flattenNp{\psi} \in \C(x)$. Clearly, $\psi$ is \continuous as
  well, so it is equivalent to $\flattenNp{\psi} $ by
  Proposition~\ref{prop:maincontinuoustwo}. It follows that $\phi$ is
  equivalent to $\flattenNp{\psi} \in \C(x)$.
\end{proof}
  A fragment of the modal $\mu$-calculus is a subset of $\Lmu$.  For
  an infinite regular cardinal $\kappa$, we let $\CC[\kappa](x)$ be
  the set of \kcontinuous formulas $\phi(x) \in \Lmu$, cf.
  Definition~\ref{def:kcontinuous}.
  We say that a fragment $\mathcal{F}$ of the modal $\mu$-calculus
  \emph{is determined by a continuity condition} if, for some infinite
  regular cardinal $\kappa$, $\mathcal{F} = \CC[\kappa](x)$.
Combining
the main result of \cite{Fontaine08} and Theorem~\ref{thm:mainone},
we immediately obtain the following result.
\begin{thm}
  \label{thm:twofragments}
  There are only two fragments of the modal $\mu$-calculus determined
  by continuity conditions: the fragment $\C[0](x)$ and the fragment
  $\C(x)$.
\end{thm}

\begin{thm}
  The following problem is decidable: given a formula
  $\phi(x) \in \Lmu$, is $\phi(x)$ \continuous for some regular
  cardinal $\kappa$?
\end{thm}
\begin{proof}
  From what has been exposed above, $\phi$ is \continuous if and only
  if it equivalent to the formula $\phi' \in \C(x)$, where
  $\phi' = \flatten{\boxingNp[x]{\phi}}$. It is then enough to observe that
  there are effective processes to construct the formula $\phi'$ and
  to check whether $\phi$ is equivalent to $\phi'$.
\end{proof}


\section{On $p$-definability}
\label{sec:models}

We collect in this section some technical results, mainly on relating
different types of submodels via formulas, that we shall use later to
prove two main results on closure ordinals of the modal
$\mu$-calculus, Theorem~\ref{thm:omegaonemonomodal} and
Theorem~\ref{thm:closuresum}.

We start recalling the usual notion of Kripke frame (hereinafter
referred to as frame).  An \emph{$\Act$-frame} (or simply, a frame, if
$\Act$ is understood) is a pair $\FF = \langle \UF,\Ras\rangle$ where
$\UF$ is a set and $R_{a} \subseteq \UF \times \UF$, for each
$a \in \Act$ -- in other words, a frame is a model without a valuation
of propositional variables.  If $v : Prop \rto P(\UF)$ is a valuation,
then we denote by $\FF_{v}$ the model $\tuple{\FF,v}$.
The complex algebra $\FF^{\sharp}$ of a frame $\FF$ is the Boolean
algebra of subsets of $\UF$ endowed with (the interpretation of) the
modal operators $\pos[a]_{\FF^{\sharp}}$, $a \in \Act$, defined by
\begin{align*}
  \pos[a]_{\FF^{\sharp}}(S) & := \set{s \in \UF \mid \exists s' \in S
    \tst sR_{a} s'}, \quad \text{for $S \subseteq \UF$.}
\end{align*}
We consider next two frames $\FF$ and $\GG$ such that
$\U\GG \subseteq \U\FF$. $\FF$ and $\GG$ might have different sets of
actions: say that $\FF$ is an $A$-frame, $\GG$ is a $B$-frame, while
we do not suppose that $A = B$. To ease the reading, we let
$F := \U\FF$ and $G := \U\GG$, so $G \subseteq F$.

The following definition formalizes the idea that each modal operator
$\pos[b]$ of the algebra $\G^{\sharp}$ is described using a term of
the algebra $\FF^{\sharp}$.
\begin{defi}
  Let 
  $\Psi = \set{\psi_{b} \in \Lmu[p,q] \mid b \in B}$ be a collection
  of formulas containing only the free variables $p,q$ in positive
  position. If $\FF$ and $\GG$ are frames as above, then we say that
  \emph{$\GG$ is $\y$-defined in $\FF$ by $\Psi$} if, for each
  $b \in B$ and each $S \subseteq F$,
  \begin{align*}
    \pos[b]_{\complex{\GG}}(G \cap S) & =
    \eval[\FF_{\substitute{G/\y,S/\z}}]{\psi_{b}(\y,\z)}\,.
   \end{align*}
\end{defi}
Above $\substitute{G/\y,S/\z}$ is the valuation that sends $\y$ to $G$
and $\z$ to $S$ (and, say, any other propositional variable to
$\emptyset$). In this sense, $\FF_{\substitute{G/\y,S/\z}}$ denotes
the model $\tuple{\FF,\substitute{G/\y,S/\z}}$.

\begin{exa}
  \label{ex:inducedsubframe}
  Suppose that $\GG$ is a \emph{subframe} of
  $\FF=\langle F,\set{R_{a}\mid a \in A}\rangle$, by which we mean
  that $A = B$, $\GG = \langle G,\set{R'_{a}\mid a \in A}\rangle$ with
  $R'_{a} = R_{a} \cap \,G \times G$, for each $a \in A$. Then
  $\GG$ is $p$-defined in $\FF$ by the collection of formulas
  $\set{p \land \pos[a](p \land q) \mid a \in A}$.
  \MYendexample{ex:inducedsubframe}
\end{exa}
\correction[]{Any of the next two examples illustrating the notion of
  $p$-definability}{The two examples we present below illustrate the
  notion of $p$-definability. Moreover, they both} shall allow (in
conjunction with Proposition~\ref{prop:translation}) to transfer
results from a bimodal setting (that is, when $\card(Act) = 2$) to a
monomodal one ($\card(Act) = 1$). In particular, \correction[]{we
  shall use one of these examples }{ the second example shall be used
  to prove Theorem~\ref{thm:omegaonemonomodal}.}

\correction[]{We let therefore}{In the following  $B :=\set{h,v}$ and $A$
is a singleton.}  The choice of the letters is suggested by the
construction in Section~\ref{sec:omega} where the actions $h$ and $v$
are interpreted respectively as horizontal and vertical
transitions. 
\begin{exa}
  \label{ex:referee}
  We are thankful to an anonymous referee for suggesting the following
  construction.  Given a bimodal frame $\G$, we define a monomodal
  frame $\FF$ on the disjoint union of the sets $\U{\G}$ and $R_{v}$
  by letting the accessibility relation be as follows:
  \begin{align*}
    x \step y\,,& \quad\text{ when }\quad x \step[h] y\,, \\
    x \step (x,y)\tand (x,y) \step y \,,& \quad\text{ when }\quad x \step[v] y\,\,.
  \end{align*}
  Clearly $\U{\G}$ embeds into $\U{\FF}$. By identifying $\U{G}$
  with its image in $\U{\FF}$,
  $\G$ is $p$-defined in $\FF$ by $\Psi = \set{\psi_{v},\psi_{h}}$,
  where
  \begin{align*}
    \psi_{h}(p,q) & = \y \land \pos(\y \land
    \z)\,,\\
    \psi_{v}(p,q) & = \y \land \pos(\neg \y \land 
    \pos(\y \land q))\,.
    \tag*{\MYendexample{ex:referee}}
  \end{align*}
\end{exa}

\begin{exa}[
  Thomason's coding of bimodal logic into monomodal logic]
  \label{ex:Thomason}
  In \cite{Thomason1980}, see also \cite[Section 4]{KrachtWolter97},
  Thomason constructs:
  \begin{enumerate}[label=(\roman*)]
  \item a monomodal formula $\simul{\phi}$, for each (fixed-point
    free) bimodal formula $\phi$; 
  \item a monomodal model $\simul{\M}$ and an injective function
    $(-)^{\circ} : \U{\M} \rto \U{\simul{\M}}$, for each bimodal model
    $\M$.
  \end{enumerate}
  These data have
  the following property:
  \begin{fact}
    \label{fact:thomason}
    For each
    $s \in \UM$, $\M,s \forces \phi$ if and only if
    $\simul{M},\wcirc{s} \forces \simul{\phi}$.
  \end{fact}
  We recall how $\simul{\M}$ is defined: for a $\set{h,v}$-model $\M$,
  $\simul{\M}$ is the monomodal model with
  $\U{\simul{\M}} = \UM\times \set{h,v} \dunion \set{\pit}$, such that
  $v(x,i) = v(x)$ and whose accessibility relation $R$ is described as
  follows:
  \begin{align*}
    (x,h) \step (y,h) \,,& \quad\text{ when }\quad x \step[h] y\,, \\
    (x,v) \step (y,v) \,,& \quad\text{ when }\quad x \step[v] y\,, \\
    (x,v) \step (x,h) \,,& \quad (x,h) \step (x,v)\,,\quad\text{ and }
    (x,h)\step\pit\,, &
  \end{align*}
  for each $x,y \in \UM$.  Since the function sending $x \in \UM$ to
  $x^{\circ } := (x,h) \in \U{\simul{\M}}$ is injective, we can
  identify $\UM$ with a subset of $\U{\simul{\M}}$.  Call $\NN$ the
  image of $\M$ within $\U{\simul{\M}}$, call $\G$ the underlying
  frame of $\NN$ and $\FF$ the underlying frame of
  $\U{\simul{\M}}$. Fact~\ref{fact:thomason} relies on $\G$ being
  $p$-defined in $\FF$ by $\Psi = \set{\psi_{h},\psi_{v}}$, where
 \begin{align*}
   \psi_{h}(p,q) & = \y \land \pos(\y \land
   \z)\,,\\
   \psi_{v}(p,q) & = \y \land \pos(\neg \y \land \pos(\neg \y \land
   \pos(\y \land q)))\,.
 \end{align*}
 The reader has remarked the similarity with the previous example.
 Thomason's construction is slightly more subtle: by adding the pit
 $\pit$ to $\M^{sim}$ and transitions as in the third line of the
 above display, the image of $\M$ under the embedding becomes
 definable by the formula $\pos\nec\bot$.  Consequently, the monomodal
 formula $\simul{\phi}$ does not contain $p$ as an additional
 propositional variable.  \MYendexample{ex:Thomason}
\end{exa}

We tackle next the proof of the main technical result of this section,
Proposition~\ref{prop:translation}. \rephrase{The Proposition}{This
  proposition} allows lifting standard simulation results (such as
Thomason's one) from modal logic to the modal $\mu$-calculus.
\begin{defi}
  \label{def:translation}
  Let $p \not \in Prop$ be a fresh variable and let $\Psi :=
  \set{\psi_{b} \in \Lmu[p,q]\mid b \in B}$.
  The formula $\tr(\phi)$ is defined by induction as follows:
  \begin{align*}
    \tr(y) & :=  \y \land y
    & 
    \tr(\neg y) & :=  \y \land  \neg y \\
    \tr(\bot) & :=  \bot
    & 
    \tr(\top) & :=  \y \\
    \tr(\psi_{0} @ \psi_{1}) & \rightside{:= \tr(\psi_{0}) @
      \tr(\psi_{1})\,, \quad @ \in \set{\land,\vee}}    \\
    \tr(\pos[b] \psi)  & := \psi_{b}[\tr(\psi)/\z] \\
    \tr(\nec[b] \psi)  & := \smallRHS{\y \land
      \psi^{op}_{b}[\tr(\psi)/\z]} \\
    \tr(\mu_{z}.\psi) & := \mu_{z}.\tr(\psi)
    &
    \tr(\nu_{z}.\psi) & := \nu_{z}.\tr(\psi)\,.
  \end{align*}
\end{defi}
In the above definition, $\psi_{b}^{op}$ is a formula dual to
$\psi_{b}$, thus semantically behaving as $\neg\psi_{b}[\neg
\z/\z]$. We need this since in the grammar \eqref{grammar:mucalculus}
we allowed negation only on propositional variables.

Aiming at a proof of
the next Proposition, let us introduce/recall
some notation: we let $\pi : P(F) \rto P(G)$ be defined by
$\pi(S) := S \cap G$; if $v : Prop \rto P(F)$, then
$\pi \circ v : Prop \rto P(G)$ is the valuation in $G$ such that
$(\pi \circ v)(y) := G \cap v(y)$, for each $y \in Prop$.
\begin{prop}
  \label{prop:translation}
  Let $p,\Psi$, and $\tr$ be as in Definition~\ref{def:translation}.
  If $\GG$ is
  $\y$-defined in $\FF$ by $\Psi$, then, for each valuation
  $v : Prop \rto P(F)$,
  \begin{align}
    \label{eq:tr}
    \evalG{\phi} & = \evalF{\tr(\phi)}\,.
  \end{align}
\end{prop}
\begin{rem}
  \label{remark:cDiagram}
  For a formula $\phi$, let us denote by
  $\eval[\FF\substitute{\pair{\y}{G}}]{\tr(\phi)}$ the mapping from
  $P(F)^{Prop}$ to $P(F)$ sending a valuation $v \in P(F)^{Prop}$
  to $\evalF{\tr(\phi)} \in P(F)$; let us denote by $\eval[\GG]{\phi}$
  the mapping sending a valuation $v' \in P(G)^{Prop}$ to
  $\eval[\GG_{v'}]{\phi} \in P(G)$.
  The statement of Proposition~\ref{prop:translation} implies that
  $\eval[\FF\substitute{G/\y}]{\tr(\phi)}$ takes values in $P(G)$ and,
  moreover, that the following diagram commutes:
   $$
   \xymatrix{
     P(F)^{Prop} \ar[d]^{\pi \circ \_}\ar[rrd]^{\qquad\eval[\FF\substitute{G/\y}]{\tr(\phi)}} \\
     P(G)^{Prop}\ar[rr]_{\eval[\GG]{\phi}} & &
     P(G)\,\,.
   }
   $$
\end{rem}
\begin{proof}[Proof of Proposition~\ref{prop:translation}]
  The proof that equation~\eqref{eq:tr} holds is by induction on
  formulas.  The basic cases are treated below:
  \begin{align*}
    \evalF{\tr(y)} & = \evalF{\y \land y}  = G \cap v(y) = \evalG{y}\,, \\
    \evalF{\tr(\neg y)} & = \evalF{\y \land \neg y}  = G \cap v(\neg y) = G
    \cap \complement{v(y)}\\
    & = G
    \cap \complement{(G \cap v(y))}  =  \evalG{\neg y}\,, \\
    \evalF{\tr(\bot)} & = \evalF{\bot} = \emptyset = \evalG{\bot}\,, \\
    \evalF{\tr(\top)} & = \evalF{p} = G = \evalG{\top}\,.
  \end{align*}
  For formulas of the form $\psi_{0}@\psi_{1}$ with
  $@ \in \set{\land,\vee}$, the result is immediate by induction. We
  give below explicit computations for formulas whose main logical
  connector is a modal operator: 
  \begin{align*}
    \evalF{\tr(\pos[b]\psi)}
    & = \evalF{\psi_{b}[\tr(\psi)/\z]} \\
    & =  \evalF[\pair{\z}{\evalF{\tr(\psi)}}]{\psi_{b}} \\
    & 
    =  \evalG[\pair{\z}{\evalG{\psi}}]{\pos[b]\z}  = \evalG{\pos[b]\psi},\\
    \evalF{\tr(\nec[b] \psi)}  & = \evalF{\y \land \neg \psi_{b}[\neg
      \tr(\psi)/\z]} \\
    & =  G \cap \complement{(\,\evalF{\psi_{b}[\neg \tr(\psi)/\z]}\,)} \\
    \tag*{with $S = \evalF{\tr(\psi)} = \evalG{\phi}$}
    & =  G \cap  \complement{(\,\evalF[\pair{\z}{\complement{S}}]{\psi_{b}}\,)} \\
    & =  G \cap \complement{(\,\evalG[G \cap \complement{S}/\z]{\pos[b]\z}\,)}
    \\
    & =  \evalG[\pair{\z}{S}]{\neg \pos[b]\neg\z} \\
    & =  \evalG[\pair{\z}{S}]{\nec[b]\z} \\
    & =  \evalG[\pair{\z}{\evalG{\psi}}]{\nec[b]\z} \\
    & =  \evalG{\nec[b]\psi} \,.
\end{align*}
We finally consider least and the greatest \FP formulas of the form
$\mu_{z}.\phi$ and $\nu_{z}.\phi$. Consider the two functions defined
by
\begin{align*}
  f(S) & := \evalF[\pair{z}{S}]{\tr(\phi)}  \quad \tand \quad 
  g(T) := \evalG[\pair{z}{T}]{\phi}\,
\end{align*}
and remark firstly their typing, that is we have $f : P(F) \rto P(F)$ and
$g : P(G) \rto P(G)$. Since by the inductive hypothesis we have
\begin{align*}
   \eval[\FF_{w}\substitute{\pair{\y}{G}}]{\tr(\phi)}  & 
   = \eval[\GG_{\pi \circ w}]{\phi}
 \end{align*}
 for each valuation $w$, this in particular holds for the valuation
 $v[\pair{z}{S}]$, with $S \subseteq F$; that is, we have
 \begin{align}
   \label{eq:fandg}
   f(S)  &= g(S \cap G)\,,
 \end{align}
 for each $S \subseteq F$.
 Let us denote by $\Pref_{h}$ the set of prefixed-points of a monotone
 function $h$ and by $\lfp.h$ its least element.\footnote{We prefer to
   use here the notation $\lfp$ in place of $\mu$ so to reserve the
   symbol $\mu$ for the syntax and to emphasize the gap
   between semantics and syntax that we are trying to fill.} 
 It immediately follows from
 equation~\eqref{eq:fandg} that $\Pref_{g}$ is included in $\Pref_{f}$
 and that $S \in \Pref_{f}$ implies $\pi(S) \in \Pref_{g}$. Therefore
 the inclusion of $\Pref_{g}$ is into $\Pref_{f}$ has $\pi$ as 
 an \Uadjoint, so it is a \Ladjoint
 and therefore (as usual for \Ladjoint{s})
 it preserves the least
 element: $\lfp.g = \lfp.f$. We obtain
 \begin{align*}
   \evalF{\tr(\mu_{z}.\psi)}
   & = \lfp.f
   = \lfp.g = \evalG{\mu_{z}.\psi}\,.
 \end{align*}
 For the greatest fixed-point, denote by $\Post_{h}$ the set of
 postfixed-points of some monotone function $h$ and by $\gfp.h$ its
 greatest element. Using equation~\ref{eq:fandg}, observe that
 $S \subseteq f(S)$ implies $S \subseteq G$. It immediately follows
 that $\Post_{f} = \Post_{g}$, so
 \begin{align*}
   \evalF{\tr(\mu_{z}.\psi)}
   & = \gfp.f 
   = \gfp.g = \evalG{\mu_{z}.\psi}\,.
 \end{align*}
 \ENDOFPROOF[Proposition~\ref{prop:translation}]
\end{proof}

It has been easier for us to expose the proof of
Proposition~\ref{prop:translation} using frames.
Next, we recast our previous observations using models, for the
particular cases of submodels (Example~\ref{ex:inducedsubframe}) and
of bimodal models (Examples~\ref{ex:referee} and~\ref{ex:Thomason}).

If $\M = \langle \UM,\RM,v \rangle$
and $\NN = \langle \U{\NN},\RN,v_\NN \rangle$ are models, then we say
that $\NN$ is a \emph{submodel} of $\M$ if $\U{\NN}$ is a subset of
$\U{\M}$ and, for each $y \in Prop$ and each $a \in Act$,
\begin{align*}
  v_{\NN}(y) & = v_{\M}(y) \cap \U{\NN}\tand 
    R_{a}^{\NN} = \U{\NN}\times\U{\NN} \cap R_{a}^{\M}.
\end{align*}
Thus, $\NN$ is a submodel of $\M$ if and only if, for some frame $\FF$,
for a valuation $v : Prop \rto P(\U{\FF})$, and for a subframe $\GG$
of $\FF$, $\M = \FF_{v}$ and $\NN = \GG_{\pi \circ v}$.  Every subset
$S$ of $\UM$ induces the submodel $\M\restr{S}$ of $\M$ defined as
follows:
\begin{align}
  \label{eq:defindsubmodel}
  \M\restr{S} & := \langle S, \set{R_{a} \cap S \times S \mid a \in
      Act },v' \rangle
\end{align}
where $v'(y) = v(y) \cap S$, is a submodel of $\M$ and it is called
the \emph{submodel of $\M$ induced} by $S$.  We write $\trp(\phi)$ in
place of $\tr(\phi)$ if $\Psi$ is the collection of formulas given in
Example~\ref{ex:inducedsubframe}.  Proposition~\ref{prop:translation}
instantiates then to models and submodels as follows:
\begin{prop}
  \label{prop:trsubmodels}
  For each formula $\phi \in \Lmu$, the formula $\trp(\phi) \in \Lmu$
  (which contains $\y$ as a new propositional variable) has the
  following property: for each model $\M$, each subset
  $S \subseteq \UM$, and each $s \in \UM$, 
  \begin{align*}
    \M[\pair{\y}{S}],s \models \trp(\phi)
    & \; \tiff \; s \in S \text{ and } \M\restr{S}, s \models \phi \,.
  \end{align*}
\end{prop}
A subset $S$ of $\UM$ is \emph{closed} if $s\in S$ and $sR_a s'$ imply
$s'\in S$, for every $a\in \Act$.  A submodel $\NN$ of $\M$ is
\emph{closed} if $\U{\NN}$ is a closed subset of $\UM$.  The attentive
reader might have already observed that if $S$ is a closed subset of
$\M$, then the statement of Proposition~\ref{prop:trsubmodels} holds
with the simpler $\y \land \phi$ in place of the recursively defined
$\trp(\phi)$.

Let us fix $\Psi$ from one of Example~\ref{ex:referee}
or~\ref{ex:Thomason}.  The translating function $\tr$ has now the
following properties:
\begin{enumerate}[label=(\roman*)]
\item it associates to each bimodal formula $\phi$ of the modal
  $\mu$-calculus a monomodal formula $\tr(\phi)$ of the modal
  $\mu$-calculus,
\item the formula $\tr(\phi)$ contains a new propositional variable,
\item the formula $\tr(\phi)$ belongs to $\C(x)$ if $\phi$
  does.
\end{enumerate}
Moreover, in case $\Psi$ comes from Example~\ref{ex:Thomason}, then
(ii) can be strengthened to the stament that $\tr(\phi)$ has exactly
the same propositional variables as $\phi$.
Proposition~\ref{prop:translation} then yields the following result.

\begin{prop}
  For each bimodal model $\M$ there is a monomodal model $\simul{\M}$
  and an injective function
  $\circfunction : \UM \rto \U{\simul{\M}}$ such that, for each
  $s \in \UM$, $\M,s \forces \phi$ if and only if
  $\simul{\M}[\pair{p}{S}],\wcirc{s} \forces \tr(\phi)$, where
  $S$ is the image of $\UM$ under the injective function
  $\circfunction$.
\end{prop}

Proposition~\ref{prop:translation} also yields the following result,
needed to transfer results on closure ordinals:
\begin{prop}
  \label{prop:inducedsubmodel}
  \label{prop:induced&closedsubmodelalpha}
  Let $\phi \in \Lmu$
  with $x$ occurring positively in $\phi$.
  \begin{enumerate}[label=(\roman*)]
  \item  If $\M$ is a model and $S \subseteq \UM$,    then
  \begin{align*}
    \trp(\phi)_{\variant{\pair{p}{S}}}^{\alpha}(\emptyset) &=
    \phi_{\M\restr{S}}^{\alpha}(\emptyset)\,.
  \end{align*}
  \item  If $\M$ is a model and $S \subseteq \UM$ is closed,    then
  \begin{align*}
    (p \land \phi)_{\variant{\pair{p}{S}}}^{\alpha}(\emptyset) &=
    \phi_{\M\restr{S}}^{\alpha}(\emptyset)\,.
  \end{align*}
\item If $\M$ is a bimodal model and both $\Psi$ and the construction
  $\simul{\M}$ come from one of the Examples~\ref{ex:referee}
  or~\ref{ex:Thomason}, then
  \begin{align*}
    \tr(\phi)_{\simul{\M}\substitute{\pair{p}{S}}}^{\alpha}(\emptyset)
    &= \wcirc{[\,\phi_{\M}^{\alpha}(\emptyset)\,]}\,,
  \end{align*}
  where $S$ is the image of $\UM$ under the injective function
  $\circfunction$.
  \end{enumerate} 
\end{prop}
\begin{proof}
  Let $\FF$, $\GG$, $F$, $G$ and $v : Prop \rto P(F)$ be as in the
  statement of Proposition~\ref{prop:translation}. If $S$ is a
  subset of $G$, then 
  \begin{align*}
    \tr(\phi)_{\FF_{v}[\pair{p}{G}]}(S) & =
    \eval[\FF_{v}\substitute{\pair{p}{G}}\substitute{\pair{x}{S}}]{\tr(\phi)}
    =
    \eval[\FF_{v}\substitute{\pair{x}{S}}\substitute{\pair{p}{G}}]{\tr(\phi)}
    \\
    & = \eval[\GG_{\pi \circ v}\substitute{\pair{x}{S}}]{\phi} 
    = \phi_{\GG_{\pi \circ v}}(S) \,.
  \end{align*}
  Then, by induction, we easily derive
  \begin{align*}
    (\tr(\phi)_{\FF_{v}[\pair{p}{G}]})^{\alpha}(\emptyset)
    & = \phi_{\GG_{\pi \circ v}}^{\alpha}(\emptyset)\,,
  \end{align*}
  for each ordinal $\alpha$.  The three statements above follow
  considering Examples~\ref{ex:inducedsubframe}, \ref{ex:referee}, and
  \ref{ex:Thomason}.
\end{proof}

Finally, consider again Example~\ref{ex:Thomason} and formulas
(resp. models) $\phi'$ (resp. $\M'$) defined by
\begin{align*}
  \phi' & := \tr(\psi)[\,\pos\nec\bot/p\,]\,, & \M' & :=
  \simul{\M}[\,\pair{p}{\eval{\pos\nec\bot/p}}\,]\,.
\end{align*}
Let us identify the injective function
$\circfunction: \UM \rto \U{\simul{\M}} $ with an inclusion (so that,
instead of embedding $\M$ into $\simul{M}$, we are actually extending
it into some bigger model).  We derive henceforth the following
simpler statement that we shall use in the next section to argue that
$\omega_{1}$ is the closure ordinal of a monomodal formula.  In the
statement the role of the special variable $p$ is not transparent
anymore.
\begin{prop}
  \label{prop:bimodaltomonomodal}
  For each bimodal formula $\phi$ there is a monomodal formula $\phi'$
  (with the same free variables of $\phi$) such that if
  $\phi \in \C(x)$, then $\phi' \in \C(x)$, and with the following
  property: for each bimodal model $\M$ there is a monomodal model
  $\M'$ (that does not depend on $\phi$) 
  such that,
  \begin{enumerate}[label=(\roman*)]
  \item $\UM \subseteq \U{\M'}$,
  \item $\M,s \forces \phi$ if and only if
    $\M',s \forces \phi'$, for each $s \in \UM$,
  \item $(\phi'_{\M'})^{\alpha}(\emptyset) =
      \phi_{\M}^{\alpha}(\emptyset)$, for each ordinal $\alpha$.
  \end{enumerate}
\end{prop}


\section{An uncountable closure ordinal}
\label{sec:closureordinals}


In this section we firstly formally define the notion of closure
ordinal, present some tools required later, here and in the next
section, and then we prove that $\omega_1$, the least uncountable
ordinal, is a closure ordinal of a formula of the modal
$\mu$-calculus.  We firstly prove it in a bimodal setting and then,
using the tools developed in the previous section, we argue that
$\omega_{1}$ is also the closure ordinal of a monomodal $\mu$-formula.


  For a formula $\phi(x)$ of the modal $\mu$-calculus and a Kripke
  model $\M$, let $\closure[\M]{\phi}$ be the least ordinal $\beta$
  for which $\phiM^{\beta}(\emptyset) =
  \phiM^{\beta+1}(\emptyset)$. 
  Recall from Definition~\ref{def:approximants} that we say that $\phiM$
  converges to its least fixed-point in exactly $\alpha$ steps when
  $\closure[\M]{\phi} = \alpha$.  
\begin{defi}
  \label{def:closureOrdinal}
  Let $\phi(x)$ be a formula of the modal $\mu$-calculus. We say that
  an ordinal $\alpha$ is the \emph{closure ordinal} of $\phi$ (and
  write $\closure{\phi} = \alpha$) if, for each model $\M$, the
  function $\phi_{\M}$ converges to its least fixed-point in at most
  $\alpha$ steps, and there exists a model $\M$ in which $\phi_{\M}$
  converges to its least fixed-point in exactly $\alpha$ steps.
\end{defi}

Elsewhere in the literature, see e.g. \cite{AfshariL13}, the closure
ordinal of a formula $\phi(x)$ w.r.t. a class of models $\clK$ is
defined as the supremum of the ordinals $\closure[\M]{\phi}$ for
$\M \in \clK$.
If $\clK$ is the class of Kripke models, then this definition coincides with the one
given above. This is a consequence of the class of Kripke models
being closed
  under disjoint unions: 
  consider a family $\set{\M_{i} \mid i \in I}$ such that
  $\alpha = \sup \set{\closure[\M_{i}]{\phi}}$;
  then the disjoint union $\dUnion_{i \in I} \U{\M_{i}}$ carries a
  canonical structure of a Kripke model, call it $\M$, and it is
  easily seen that $\closure[\M]{\phi} = \alpha$.  

\smallskip

The notions of \emph{closure ordinal of a formula on a structure} and
of \emph{closure ordinal of a structure} appear in the monograph
\cite[Chapter 2B]{EIAS}. 
The notion of closure ordinal presented here is on the other hand
strictly related to \emph{global} inductive definability, see
\cite{BM1978}. Indeed, it is well-known that each \FP-free modal
formula $\psi$ can be transformed into some equivalent first order
logic sentence $\sttr{\psi}$, known as the \emph{standard translation}
of $\psi$. The formula $\sttr{\psi}$ contains $y$ as the only
free-variable and is related to $\psi$ by the equivalence
$\M,s \forces \psi$ if and only if $\M \models \sttr{\psi}(s)$, where
$\M$ is considered as a relational structure for first-order
logic. The closure ordinal of a \FP-free modal formula $\phi(x)$, as
defined here and when it exists, coincides with the global closure
ordinal of the first-order inductive definition given by
$\sttr{\phi(x)}$.

Let us recall that formulas may have no closure ordinal. For example
$\phi(x) := \nec x$ has no closure ordinal.  Indeed, it is not
difficult to construct, for each ordinal $\alpha$, a model
$\M_{\alpha}$ such that $\phi_{\M}^{\alpha}(\emptyset)$ is strictly
included in $\phi_{\M}^{\alpha+1}(\emptyset)$. We collect with the
following Proposition the observations developed in the course of the
paper that are relevant to closure ordinals.
\begin{prop}\label{prop:w1upperbound}
  If a formula $\phi(x)$ belongs to the syntactic fragment $\C(x)$,
  then it has a closure ordinal $\closure{\phi(x)}$ and
  $\omega_{1}$ is an upper bound for $\closure{\phi(x)}$.  
\end{prop}
\begin{proof}
  The formula $\phi$ belongs to the syntactic fragment $\C(x)$, thus
  it is \continuous[$\aleph_{1}$] and, for every model $\M$,
  $\phi_{\M}$ is \continuous[$\aleph_{1}$]. It follows then from
  Proposition~\ref{prop:convkcontinuous} that
  $\phi_{\M}$ converges to its least fixed-point in at most
  $\omega_{1}$ steps.  Therefore, for such $\phi$,
  $\sup \set{\closure[\M]{\phi} \mid \M \text{ a Kripke model}} \leq
  \omega_{1}$. As we have seen at the beginning of this section, there
  exists a model $\M$ such that
  $\closure[\M]{\phi} = \sup \set{\closure[\M]{\phi} \mid \M \text{ a
      Kripke model}}$. 
\end{proof}


The following Lemma will be useful in the next section, when we shall
show that closure ordinals of the modal $\mu$-calculus are closed
under ordinal sum.
\begin{lem}
  \label{prop:totallfp}
  \label{prop:lfp.define.ordinals}
  \label{lemma:totallfp}
  \label{lemma:lfp.define.ordinals}
  Let $\alpha \neq 0$ be a closure ordinal of the modal
  $\mu$-calculus.  Among the formulas that have $\alpha$ as its
  closure ordinal there exists one formula $\phi(x)$ such that
  $\mu_{x}. {\phi}(x)$ is total in some model $\M$ where the
  convergence occurs in exactly $\alpha$ steps, that is,
  \begin{align*}
    \U\M&= \eval[\M]{\mu_{x}.\phi(x) }=
    \phiM^\alpha(\emptyset)\neq \phiM^{\alpha'}(\emptyset)\,,
    \tag*{ for every $\alpha'<\alpha$.}
  \end{align*}
\end{lem}
\begin{proof}
  For a formula $\psi(x)$, let $(\mu_{x}.\psi(x))^{op}$ be a formula
  semantically equivalent to the negation of $\mu_{x}.\psi(x)$ and
  define then
  \begin{align*}
    \phi(x) & := (\mu_{x}.\psi(x))^{op} \vee \psi(\,x \land \mu_{x}.\psi(x)\,)\,.
  \end{align*}
  Observe that $\phi(x)$ is not \wnamed, yet this will not be a 
  concern here.
  \rephrase{Let us verify that, for each ordinal $\gamma \geq 1$, we have
    \begin{align}
    \phi_{\M}^{\gamma}(\emptyset)
    & = (\mu.\psi_{\M}) \impl \psi_{\M}^{\gamma}(\emptyset)\,.
  \end{align}
  This is clear if $\gamma = 1$. Assuming this holds for $\gamma$,
  then
  \begin{align*}
    \phi_{\M}^{\gamma +1}(\emptyset) & = \mu.\psi_{\M} \impl
    \psi_{\M}(\mu.\psi_{\M} \impl \psi_{\M}^{\gamma}(\emptyset)
    \cap \mu.\psi_{\M}) \\
    & = \mu.\psi_{\M} \impl \psi_{\M}(\psi_{\M}^{\gamma}(\emptyset)
    \cap \mu.\psi_{\M})  \\
    & = \mu.\psi_{\M} \impl \psi_{\M}(\psi_{\M}^{\gamma}(\emptyset) ) = \mu.\psi_{\M}
    \impl \psi_{\M}^{\gamma + 1}(\emptyset)\,.
  \end{align*}
  The inductive step is obvious. Clearly, $\mu.\phi_{\M}$ is always
  the total set $\UM$.}{%
  For the sake of readability, let $\mu \eqdef \lfp.\psiM$. We verify
  next that
  \begin{align}
    \label{eq:characterisationIteration}
    \phi_{\M}^{\gamma}(\emptyset)
    & = \mu \impl \psi_{\M}^{\gamma}(\emptyset)\,,
    \qquad\text{for each ordinal $\gamma \geq 1$.}
  \end{align}
  The symbol $\impl$ used above stands for the Heyting implication of
  the Boolean algebra $P(\U{\M})$.
  Equation~\eqref{eq:characterisationIteration} clearly holds if
  $\gamma = 1$. Assuming the equation holds for $\gamma$, then
  \begin{align*}
    \phi_{\M}^{\gamma +1}(\emptyset) & = \mu\impl
    \psi_{\M}(\,(\mu\impl \psi_{\M}^{\gamma}(\emptyset))\,
    \cap \,\mu\,) \\
    & = \mu\impl \psi_{\M}(\,\psi_{\M}^{\gamma}(\emptyset)
    \cap \mu\,)  \\
    & = \mu\impl \psi_{\M}(\,\psi_{\M}^{\gamma}(\emptyset) \,)\,,
    \tag*{since
      $\psiM^{\gamma}(\emptyset) \subseteq \lfp.\psiM = \mu$,}
    \\
    & = \mu \impl \psi_{\M}^{\gamma + 1}(\,\emptyset\,)\,.
  \end{align*}
  The inductive step to a limit ordinal is obvious.  From
  equation~\eqref{eq:characterisationIteration} it follows that, for
  each $\gamma \neq 0$,
  $\phiM^{\gamma+1}(\emptyset) \subseteq \phiM^{\gamma}(\emptyset)$ if
  and only if
  $\psiM^{\gamma+1}(\emptyset) \subseteq \psiM^{\gamma}(\emptyset)$,
  so $\closure[\M]{\phi} = \closure[\M]{\psi}$ provided that
  $\closure[\M]{\psi} > 0$.
  Finally,
  $\eval[\M]{\mu_{x}.\phi(x)} =\lfp.\phi_{\M} = \mu \impl \mu = \UM$.  }
\end{proof}

\subsection{$\omega_{1}$ is a closure ordinal}
\label{sec:omega}
We are going to prove that $\omega_{1}$ is the closure ordinal of the
following bimodal formula:
\begin{align}
  \PHI(x) & := (\nu_{z}.\pos[v] x \land \pos[h] z) \vee \nec[v]
  \bot\,.
  \label{eq:defphi}
\end{align}

For the time being, consider $\Act=\sset{h,v}$; if
$\M = \langle \UM,R_{h},R_{v},v\rangle$ is a model, we think of
$R_{h}$ as a set of horizontal transitions and of $R_{v}$ as a set of
vertical transitions.
Thus, for $s \in \UM$, $\M,s \forces \PHI(x)$ if either (i) there are
no vertical transitions from $s$, or (ii) there exists an infinite
horizontal path from $s$ such that each state on this path has a
vertical transition to a state $s'$ such that $\M,s' \forces x$.

By
Proposition~\ref{prop:w1upperbound}, the formula $\PHI(x)$ has a
closure ordinal and $\closure{\PHI(x)}\leqslant \omega_{1}$.  In order
to prove that $\closure{\PHI(x)} = \omega_{1}$, we are going to
construct a model $\MO$ where
$\PHI^{\omega_{1}}_{\MO}(\emptyset) \not\subseteq
\PHI^{\alpha}_{\MO}(\emptyset)$ for each $\alpha < \omega_{1}$.

\medskip

The construction relies on a few combinatorial properties of posets
and ordinals that we recall here.
For a poset $P$ and an ordinal $\alpha$, an $\alpha$-chain in $P$ is a
subset $\set{p_{\beta} \mid \beta < \alpha} \subseteq P$, with
$p_{\beta} \leq p_{\gamma}$ whenever $\beta \leq \gamma < \alpha$.  An
$\alpha$-chain $\set{p_{\beta} \mid \beta < \alpha} \subseteq P$ is
\emph{cofinal} in $P$ if, for every $p \in P$ there exists
$\beta < \alpha$ with $p \leq p_{\beta}$.  The \emph{cofinality}
$\kappa_{P}$ of a poset $P$ is the least ordinal $\alpha$ for which
there exists an $\alpha$-chain cofinal in $P$.
Recall that an ordinal $\alpha$ might be identified with the poset
$\set{\beta \mid \text{$\beta$ is an ordinal}, \beta < \alpha}$ and so
$\kappa_{\alpha} = \omega$, whenever $\alpha$ is a countable infinite
limit ordinal; this means that, for such an $\alpha$, it is always
possible to pick an $\omega$-chain cofinal in $\alpha$.

\medskip

For a given ordinal $\alpha \leq \omega_{1}$, let
  \begin{align*}
    S_{\alpha} & := \set{(n,\beta) \mid 0 \leq n < \omega,\,
      \text{$\beta$ is an ordinal}, \beta < \alpha }
    \,.
  \end{align*}
  We define $\MO$ to be the model
  $\langle S_{\omega_{1}},R_{h},R_{v},v\rangle$ where
  $v(y)= \emptyset$, for each $y \in Prop$,
  horizontal transitions are of the form
    $(n,\beta) \step[h] (n +1,\beta)$, for each $n < \omega$ and
    each ordinal $\beta$, and 
    vertical transitions from a state $(n,\beta) \in S_{\omega_{1}}$
    are as follows:
    \begin{itemize}
    \item if $\beta = 0$, then there are no vertical transitions
      outgoing from $(n,0)$;
      
    \item if $\beta = \gamma +1$ is a successor ordinal, then the
      only vertical transitions are of the form
      $(n,\gamma + 1) \step[v] (0,\gamma)$;
    \item if $\beta$ is a countable limit ordinal distinct from $0$,
      then vertical transitions are of the form
      $(n,\beta) \step[v] (0,\beta_{n} )$, where
      the set $\set{\beta_{n} \mid n < \omega}$ is a chosen
      $\omega$-chain cofinal in $\beta$.
    \end{itemize}

\begin{lem}\label{lemma:Salphaplusone}
  For each countable ordinal $\alpha$, we have
  \begin{align*}
    \phi_{\MO}(S_{\alpha}) & = S_{\alpha + 1}.
  \end{align*}
  Consequently,
  for each ordinal $\alpha \leq \omega_{1}$, we have
    $\phi_{\MO}^{\alpha}(\emptyset) = S_{\alpha}$.
  \end{lem}
  \begin{proof}
    If $\alpha = 0$, then $S_{\alpha} = \emptyset$ and
    \begin{align*}
      \phi_{\MO}(S_{0}) & = \phi_{\MO}(\emptyset) =
      \eval[\MO\substitute{\pair{x}{\emptyset}}]{\nu_{z}.(\pos[h] z \land
        \pos[v] x) \vee \nec[v]\bot} \\ & = \eval[\MO]{\nec[v]\bot}
      = \set{(n,0) \mid n < \omega}
      = S_{1}\,.
    \end{align*}
    Consider now an ordinal $\alpha > 0$.

    \smallskip
    
    Let us argue firstly that $ S_{\alpha + 1} \subseteq \phi_{\MO}(S_{\alpha})$.
    Let $(n,\beta) \in S_{\alpha + 1}$, so
    $\beta < \alpha + 1$ implies
    $\beta \leq \alpha$.
    From $(n,\beta)$, 
    there is the infinite horizontal path
    $\set{(m,\beta) \mid n \leq m < \omega}$ and each vertex on this
    path has a vertical transition to a vertex $(0,\beta')$ with
    $\beta' < \beta \leq\alpha$, in particular
    $(0,\beta')\in S_{\alpha}$. Therefore $(n,\beta)\in \phi_{\MO}(S_{\alpha})$.

    \smallskip

    Next, we argue that the converse inclusion,
    $\phi_{\MO}(S_{\alpha}) \subseteq S_{\alpha + 1}$, holds.
    Suppose $(n,\beta) \in \phi(S_{\alpha})$. If there are no vertical
    transitions from $(n,\beta)$ then $\beta = 0$ and
    $(n,\beta) = (n,0) \in S_{1} \subseteq S_{\alpha +1}$, since
    $S_{\beta} \subseteq S_{\gamma}$ for $\beta \leq
    \gamma$. Otherwise $\beta > 0$, there is an infinite horizontal
    path from $(n,\beta)$ and each vertex on this path has a
    transition to some vertex in $S_{\alpha}$. Notice that such an
    infinite horizontal path is, necessarily, the path
    $\pi := \set{(m,\beta) \mid n \leq m < \omega}$.

    If $\beta = \gamma +1$ is a successor ordinal then the unique
    outgoing vertical transition from $(n,\beta)$ is to
    $(0,\gamma)$. Hence $(0,\gamma)\in S_{\alpha}$, thus
    $\gamma < \alpha$, $\beta = \gamma + 1< \alpha + 1$ and
    $(n,\beta)\in S_{\alpha +1}$. Otherwise $\beta$ is a limit ordinal
    distinct from $0$ and, for each $m\geqslant n$, there is a
    vertical transition $(m,\beta)\step[v] (0,\beta_{m})$ with
    $ (0,\beta_{m})\in S_{\alpha}$, so $\beta_{m} < \alpha$.  If
    $ \alpha +1 \leq \beta$, then $\alpha < \beta$, that is,
    $\alpha \in \beta$. Since the $\omega$-chain
    $\set{\beta_{k} \mid k \in\omega}$ is cofinal in $\beta$, we can
    find $k \in \omega$ such that $\alpha \leq \beta_{k}$. Since
    $\beta_{k} \leq \beta_{k'}$ for $k \leq k' \in \omega$, we can
    also suppose that $n \leq k$. But we obtain here a contradiction,
    since we mentioned before that $\beta_{m} < \alpha$ for
    $m \geq n$, in particular $\beta_{k} < \alpha$.

    The proof of the second statement is now a straightforward
    induction on the ordinal $\alpha$.
    If $\alpha = \beta + 1$ is a successor ordinal, then
    \begin{align*}
      \phi_{\MO}^{\alpha}(\emptyset) & =
      \phi_{\MO}(\phi_{\MO}^{\beta}(\emptyset)) =
      \phi_{\MO}(S_{\beta}) = S_{\beta +1}\,.
    \end{align*}
      If $\alpha$ is a limit ordinal, then
      \begin{align*}
        \phi_{\MO}^{\alpha}(\emptyset) & = \bigcup_{\beta < \alpha}
        \phi_{\MO}^{\beta}(\emptyset)
        = \bigcup_{\beta < \alpha} S_{\beta} = S_{\alpha}\,.
      \end{align*}
      \ENDOFPROOF[Lemma~\ref{lemma:Salphaplusone}]
    \end{proof}

    We conclude the section by stating its main result.
    \begin{thm}
      \label{thm:omegaone}
      The closure ordinal of $\PHI(x)$ is $\omega_{1}$.
    \end{thm}
    \begin{proof}
      As we mentioned before the formula $\PHI(x)$ has a closure
      ordinal and $\closure{\PHI(x)}\leqslant \omega_{1}$, by
      Proposition~\ref{prop:w1upperbound}.  We claim that $\PHI_{\MO}$
      converges to its \LFP in exactly $\omega_{1}$ steps, that is, we
      have
      $\PHI_{\MO}^{\omega_{1}}(\emptyset) \not\subseteq
      \PHI^{\alpha}_{\MO}(\emptyset)$ for each $\alpha < \omega_{1}$.
      Our claim is verified as follows.
      By Lemma~\ref{lemma:Salphaplusone}, the claim is equivalent to
      $S_{\omega_{1}} \not\subseteq S_{\alpha} $, for each
      $\alpha < \omega_{1}$. The latter relation holds since if
      $\alpha < \omega_{1}$, then we can find an ordinal $\beta$ with
      $\alpha < \beta < \omega_{1}$, so the states $(n,\beta)$,
      $n \geq 0$, belong to $S_{\omega_{1}}\setminus S_{\alpha}$.
    \end{proof}
  
   
    Finally, we argue that a bimodal language is not needed for
    $\omega_{1}$ to be a closure ordinal. To this goal, let $\Psi$ be
    as in Example~\ref{ex:Thomason} and let
    \begin{align*}
      \PHI' &:= \tr(\PHI)[\, \pos\nec \bot /p\,]\,, & \MO' :=
      \simul{\MO}[\,\pair{p}{\eval[\MO]{\pos\nec\bot}}\,]\,,
    \end{align*}
    where $\PHI$ is the bimodal formula defined in
    equation~\eqref{eq:defphi}. As in the statement of
    Proposition~\ref{prop:bimodaltomonomodal}, we consider $\U{\MO'}$
    as a superset of $\U{\MO}$.
  \begin{thm}
    \label{thm:omegaonemonomodal}
    The monomodal formula $\PHI'$
    has closure ordinal $\omega_{1}$.
  \end{thm}
  \begin{proof}
    Consider the statement of
    Proposition~\ref{prop:bimodaltomonomodal}.  Since the
    correspondence $\phi \mapsto \phi'$ sends formulas in $\C(x)$ to
    formulas in $\C(x)$, $\PHI'$ is \continuous[$\aleph_{1}$] and
    therefore it has a closure ordinal bounded by $\omega_{1}$. To
    argue that the closure ordinal of $\PHI'$ is equal to $\omega_{1}$
    it is enough to consider the model $\MO'$ and rely on item (iii)
    of Proposition~\ref{prop:bimodaltomonomodal}.
  \end{proof}

\section{Closure under ordinal sum.}
\label{sec:ordinalsum}

In this section we prove 
that the ordinal sum of two closure ordinals of the modal
$\mu$-calculus is again a closure ordinal of this logic, as stated in
the next theorem.

\begin{thm}
  \label{thm:closuresum}
  Suppose $\phi_0(x)$ and $\phi_1(x)$ are monomodal formulas that
  have, respectively, $\alpha$ and $\beta$ as closure ordinals. 
  Then there is a monomodal formula $\Psi(x)$, constructible from
  $\phi_{0}$ and $\phi_{1}$, whose closure ordinal is
  $\alpha + \beta$.
\end{thm}
We prove the theorem through a series of observations. With the first
one, Lemma~\ref{lemma:acceptablemodel},
we make use of the \master modality $\Master$ of the propositional
modal $\mu$-calculus. In principle, the use \master modality in the
proof of Theorem~\ref{thm:closuresum} may be avoided, at the cost of
reducing its readability.  
Given a monomodal formula $\chi$ this modality is defined as follows:
\begin{align*}
  \Master\chi \eqdef \nu_{z}.(\,\chi \land \nec z\,)\,.
\end{align*}
The master modality allows us to focus on those models of a fixed
shape since they satisfy, globally, a given formula. Indeed, the
semantics of this modality is the following:
\begin{align*}
  \M,s \forces \Master \chi & \text{ if and only if } \M, s' \forces
  \chi, \text{ for each $s'$ reachable from $s$}.
\end{align*}
In particular, if $\M$ is a tree model, then $\M \forces \chi$ 
 if and only if
$\M,r \forces \Master \chi$, where $r$ is the root of the tree.
Let us mention that the modality $\Master$ satisfies all the axioms
(reflexivity and transitivity) of the modal system \Sfour, see
e.g. \cite[\S~2.5]{Kracht1999}, and yields a deduction theorem for the
modal $\mu$-calculus, see \cite{KrachtMCR, SantocanaleFICS02}.

When $\M \forces \Master \chi$ (that is, $\M,s \forces \Master \chi$,
for each $s \in \UM$), we say that $\M$ is \emph{\chiacceptable}.
\begin{lem}
  \label{lemma:acceptablemodel}
  Let $\chi$ and $\psi(x)$ be monomodal formulas and define
  $\Psi(x)\eqdef \Master \chi \land
  \psi(x)$. 
  An ordinal $\gamma$ is the closure ordinal of the formula $\Psi(x)$
  if and only if
  (i) the formula $\psi(x)$ converges to its least fixed point in at
  most $\gamma$ steps on all the \chiacceptable models, and (ii) there
  exists an \chiacceptable model on which the formula $\psi(x)$
  converges to its least fixed point in exactly $\gamma$ steps.
\end{lem}
\begin{proof}
  If $\NN$ is an \chiacceptable model, then
  $\eval[\NN]{\Master \chi} = \U{\NN}$, so that
  $\Psi_{\NN} = \psi_{\NN}$.
  
  On the other hand, if $\M$ is any model, then the submodel of $\M$
  induced by $\eval[\M]{\Master \chi}$ is closed and
  \chiacceptable. Call $\NN$ such a submodel of $\M$. Thus, by
  Proposition~\ref{prop:induced&closedsubmodelalpha}.(ii), for any
  ordinal $\gamma \geq 0$, we have
  \begin{align}
    \label{eq:indsubmodel}
    \Psi_\M^\gamma(\emptyset)
    & =
    \psi_\NN^\gamma(\emptyset)\,.
  \end{align}
  The statement of the lemma immediately follows. 
\end{proof}

\medskip

Next, recall that we write $\trp(\phi)$ in place of $\tr(\phi)$ if
$\Psi$ is the collection of formulas given in
Example~\ref{ex:inducedsubframe}.
Let $\phi_0(x)$ and $\phi_1(x)$ be monomodal formulas as in the
statement of Theorem~\ref{thm:closuresum}. For a
variable $p$ occurring neither in ${\phi_0}$ nor in ${\phi_1}$, we define  
\begin{align}
  \label{eq:defchi}
  \chi  & \eqdef\chi_0 \land \chi_1 \text{ with }
  \chi_0 \eqdef p \vee (\, \nec \neg p \, \land \, 
  \mu_z.\phi_0(z)\,)
  \text{ and }
  \chi_1 
  \eqdef \neg p \vee  \mu_z.\trp(\phi_1(z)) \,, 
  \\
  \psi(x) & \eqdef(\, \neg p \land \phi_0(x) \,)
  \,\vee\,(\,\phiU(x)\land \nec( p \vee x) \,)\,,
  \label{eq:defpsi} \\
  \Psi(x) &\eqdef \Master \chi \land \psi(x)\,.
  \label{eq:defPsi}
\end{align}

From now on, we shall say  that a model $\NN$ is
\emph{\acceptable} if it is \chiacceptable, where $\chi$ is the
formula given in equation~\eqref{eq:defchi}.
We shall
argue that $\Psi(x)$ defined in~\eqref{eq:defPsi} has closure ordinal
$\alpha + \beta$ using Lemma~\ref{lemma:acceptablemodel}.
  
Next, we continue by studying the structure of an acceptable model
$\NN$ and how $\psi_{\NN}$ acts on it---where $\psi$ is the formula
defined in~\eqref{eq:defpsi}.  To this goal, let $\NN_{0}$ and
$\NN_{1}$ be the submodels of $\NN$ induced by $v(\neg p)$ and $v(p)$,
respectively.  To ease the reading, let $N_{0} \eqdef v(\neg p)$, and
$N_{1} \eqdef v(p)$.  A model $\NN$ is \acceptable if and only if
$N_{0}$ is a closed subset of $\U{\NN}$ (since
$\NN \forces p \vee \nec\neg p \equiv \neg p \impl \nec\neg p$) and
moreover
\begin{align*}
  N_0 & \subseteq \eval[\NN]{\,\mu_z.\phi_0(z)\,} \,, \qquad
  N_1  \subseteq \eval[\NN]{\,\mu_z.\trp(\phi_1(z))\,}\,.
\end{align*}
Let
also $\phiNz \eqdef (\phi_{0})\mphin{\NN_{0}}$ and
$\phiNu \eqdef (\phi_{1})\mphin{\NN_{1}}$, so
$\phiNz : P(N_{0}) \rto P(N_{0})$ and
$\phiNu : P(N_{1}) \rto P(N_{1})$.  We claim that $\psiN$ is of the
form
  \begin{align}
    \label{eq:psiN}
    \psiN(X) & = \phiNz(X \cap N_{0}) \cup (\phiNu(X \cap N_{1})
    \cap \DELTA(X \cap N_{0}))\,,
    \intertext{with}
    \DELTA(X) & \eqdef N_{1} \cap \nec_{\NN}(N_{0} \impl X)\,.
  \end{align}
  This is because,   for each
  $X \subseteq \U\NN$,
  \begin{align*}
    \psi_{\NN}(X) & = (\psi_{\NN}(X) \cap N_{0}) \cup (\psi_{\NN}(X)
    \cap N_{1})\,, \\
    \psi_{\NN}(X) \cap N_{0} & = \phiNz(X \cap
    N_{0})\,,  \\
    \psi_{\NN}(X) \cap N_{1} & =  \phiU\mphin{\NN_{1}}(X) \cap
    \nec_{\NN}(N_{0} \impl X)
    = N_{1} \cap \phiNu(X
    \cap N_{1}) \cap
    \nec_{\NN}(N_{0} \impl X)
    \\
    &
    = \phiNu(X
    \cap N_{1}) \cap
    N_{1} \cap \nec_{\NN}(N_{0} \impl (X \cap N_{0}))\,.
  \end{align*}
  We notice now that if $\NN$ is \acceptable, then
  \begin{align}
    \label{eq:Nz}
    N_0 & =\eval[\NN]{\mu_z.\phi_0(z)} \cap N_{0} =
    \eval[\NN_{0}]{\mu_z.\phi_0(z)} = \phiNz^{\alpha}(\emptyset)
    \intertext{and} 
    \label{eq:Nu}
    N_1 & = \eval[\NN]{\mu_z.\trp(\phi_1(z))} \cap
    N_{1} =
    \eval[\NN_{1}]{\mu_z.\phi_1(z)} = \phiNu^{\beta}(\emptyset)\,.
  \end{align}
  Observe that $\DELTA(X) = N_{1}$ whenever $N_{0} \subseteq X$ and
  therefore, using $\phiNz^{\alpha}(\emptyset) = N_{0}$, we have
   \begin{align}
     \label{eq:delta}
     \DELTA(X) & = N_{1}, \quad \text{whenever
       $X \supseteq \phiNz^{\alpha}(\emptyset)$}\,.
   \end{align}

   \begin{lem}
     \label{prop:Sumacceptablemodel}
     \label{lemma:Sumacceptablemodel}
     On every \acceptable model $\NN$ the equality
     $\psi_\NN^{\alpha + \beta}(\emptyset)=\U{\NN}$ holds and,
     consequently, the formula $\psi(x)$ converges within
     $\alpha + \beta$ steps.
   \end{lem}
   \begin{proof}
     Since $N_{0}$ is a closed subset of $\U{\NN}$, by
     Proposition~\ref{prop:induced&closedsubmodelalpha}, we have
    \begin{align}
    \label{eq:psiNz}
      \psiN^{\delta}(\emptyset) \cap N_{0} & =
      \psi_{\NN_{0}}^{\delta}(\emptyset) =
      \phiNz^{\delta}(\emptyset)\,
    \end{align}
    for each ordinal $\delta$. Consequently,
    $\psi_{\NN}^{\alpha + \gamma} (\emptyset)\cap N_{0} \supseteq
    \psi\mphin\NN^{\alpha} (\emptyset)\cap N_{0} = \phiNz^{\alpha}
    (\emptyset)$, for every ordinal $\gamma$.
   \begin{clm}
     \label{lemma:inclusionphipsi}
     The following relation  holds for every ordinal $\gamma \geq 0$:
     \begin{align}
       \label{eq:inclusion}
       \phiNu^{\gamma}(\emptyset) & \subseteq \psiN^{\alpha +
         \gamma}(\emptyset) \cap N_{1}\,.
     \end{align}
   \end{clm}
   \begin{proofofclaimNoQed}
     Clearly the relation holds for $\gamma =
     0$.  In order to prove the above inclusion, it will be enough to
     prove that it holds at a successor ordinal $\gamma
     +1$, assuming it holds at
     $\gamma$ (the inductive step to a limit ordinal is obvious).  We
     have
     \begin{align*}
       {\psiN^{\alpha+\gamma + 1}(\emptyset) \cap N_{1}} & =
       \phiNu(\psiN^{\alpha+\gamma}(\emptyset) \cap N_{1}) \cap
       \DELTA(\psiN^{\alpha + \gamma}(\emptyset) \cap N_{0})
       \\
       & = \phiNu(\psiN^{\alpha+\gamma}(\emptyset) \cap N_{1}) \cap
       \DELTA(\phiNz^{\alpha +\gamma}(\emptyset))\,, \tag*{by
         equation~\eqref{eq:psiNz},}
        \\
        & = \phiNu(\psiN^{\alpha+\gamma}(\emptyset) \cap N_{1}),
        \tag*{by equation~\eqref{eq:delta},}\\
        \tag*{by the IH,} & \supseteq
        \phiNu(\phiNu^{\gamma}(\emptyset))\,, \\
        & = \phiNu^{\gamma + 1}(\emptyset)\,. \tag*{\qedClaim}
     \end{align*}
   \end{proofofclaimNoQed}
   Therefore
   \begin{align*}
     \U{\NN}=N_0\cup N_1&= \phiNz^{\alpha}(\emptyset) \cup
     \phiNu^{\beta}(\emptyset) \tag*{using \eqref{eq:Nz} and \eqref{eq:Nu}}\\
     & \subseteq (\psiN^{\alpha + \beta}(\emptyset) \cap N_0) \cup
     (\psiN^{\alpha + \beta}(\emptyset) \cap N_1) = \psiN^{\alpha +
       \beta}(\emptyset)\,.
   \end{align*}
   This terminates the proof of Lemma~\ref{lemma:Sumacceptablemodel}.
 \end{proof}

 \begin{lem}
   \label{lemma:existenceOfModel}
   There exists an \acceptable model $\NN$ on which $\psi(x)$
   converges in exactly $\alpha + \beta$ steps.
 \end{lem}
 \begin{proof}
   Since the formulas ${\phi_0}(x)$ and ${\phi_1}(x)$ have,
   respectively, $\alpha$ and $\beta$ as closure ordinals, by
   Lemma~\ref{lemma:totallfp} there exist models
   $\M_{\gamma} = \tuple{\U{\M_{\gamma}},R_{\gamma},v_{\gamma}}$,
   $\gamma \in \set{\alpha,\beta}$, such that for every
   $\alpha'<\alpha$ and $\beta'<\beta$
   $\eval[\M_\alpha]{\mu_{x}. {\phi_0}(x) }=\U{\M_\alpha}=
   {\phi_0}\mphin{\M_\alpha}^\alpha(\emptyset)\neq {\phi_0}\mphin
   {\M_\alpha}^{\alpha'}(\emptyset)$ and
   $\eval[\M_\beta]{\mu_{x}.{\phi_1}(x)}=\U{\M_\beta}={\phi_1}\mphin
   {\M_\beta}^{\beta}(\emptyset)\neq {\phi_1}\mphin
   {\M_\beta}^{\beta'}(\emptyset)$.

   We construct now the model $\M_{\alpha + \beta}$ by making the
   disjoint union of the sets $\U{\M_{\alpha}}$ and $\U{\M_{\beta}}$,
   endowed with
   $R_{\alpha}\cup R_{\beta}\cup\set{(s, s')\mid s\in
     \U{\M_{\beta}},s'\in \U{\M_{\alpha}}}$ and the
   valuation $v$ defined by 
   $v(q) \eqdef \U{\M_{\beta}}$, if $q = p$, and
   $v(q) \eqdef v_\alpha(q) \cup v_{\beta}(q)$ otherwise.
   Let us put $\NN\eqdef M_{\alpha + \beta}$.  Observe now that
   $\M_{\alpha + \beta}$ is an \acceptable model and that
   $\DELTA(X)=\emptyset$ for every $X\subseteq\U{\NN}$ such that
   $X\cap N_0\subsetneq\phiNz^{\alpha}(\emptyset)$.  Because of this,
   the inclusion~\eqref{eq:inclusion} is actually an equality, as
   stated and proved next. 
   \begin{clm}
     \label{lemma:inclusionphipsiBISEQUALITY}
     Suppose that $\phiNz^{\delta}\eset$ is strictly included in
     $N_{0}$ for $\delta < \alpha$ and that $\DELTA(X) = \emptyset$
      whenever $X$ is a proper subset of $N_{0}$. Then, the
      inclusion~\eqref{eq:inclusion} is an equality, for each ordinal
      $\gamma \geq 0$:
      \begin{align}
        \label{eq:inclusionIsEquality}
        \phiNu^{\gamma}(\emptyset) & = \psiN^{\alpha +
          \gamma}(\emptyset) \cap N_{1}\,.
      \end{align}
    \end{clm}
    \begin{proofofclaimNoQed}
      It is enough to verify that the above equality holds for
      $\gamma = 0$. Indeed, for $\gamma > 0$, we can use the
      same 
      computations as in the proof of the claim in
      Lemma~\ref{lemma:Sumacceptablemodel}, by substituting an
      equality for the inclusion in the inductive hypothesis.

      If
      $\delta < \alpha$, then 
      \begin{align*}
        \psiN^{\delta+1}\eset \cap N_{1} & \subseteq
        \DELTA(\psiN^{\delta}\eset \cap N_{0} )
        = \DELTA(\phiNz^{\delta}\eset) = \emptyset\,,
      \end{align*}
      since by assumption $\phiNz^{\delta}\eset$ is strictly included
      in $N_{0}$.
      In particular, if $\alpha$ is a successor ordinal, we have $\psiN^{\alpha +
        \gamma}(\emptyset) \cap N_{1} = \emptyset$.
      If $\alpha$ is a limit ordinal, then
      \begin{align*}
        \psiN^{\alpha} \eset \cap N_{1} & \subseteq \bigcup_{\delta <
          \alpha} \psiN^{\delta +1}\eset  \cap N_{1} \subseteq \bigcup_{\delta <
          \alpha} \DELTA(\psiN^{\delta}\eset  \cap N_{0}) = \emptyset\,.
        \tag*{\qedClaim}
      \end{align*}
    \end{proofofclaimNoQed}

    We can then use equations~\eqref{eq:psiNz}
    and~\eqref{eq:inclusionIsEquality} to obtain
    \begin{align*}
      \psiN^{\alpha}(\emptyset) = \phiNz^{\alpha}(\emptyset)
      \supsetneqq
      \phiNz^{\delta}(\emptyset)=\psiN^{\delta}(\emptyset)
      \tand
      \psiN^{\alpha + \gamma}(\emptyset)=N_0 \cup
      \phiNu^{\gamma}(\emptyset)
    \end{align*}
    for ordinals $\gamma,\delta$ such that $\delta <\alpha$.
    Finally,
    \begin{align*}
      \psiN^{\alpha + \beta}(\emptyset) =\U{\NN} = N_0 \cup
      \phiNu^{\beta}(\emptyset) 
      \supsetneqq 
      N_0 \cup \phiNu^{\gamma}(\emptyset)= \psiN^{\alpha +
        \gamma}(\emptyset) \,, \quad\text{for $\gamma <\beta$}\,.
    \end{align*}
    This shows that $\psi$ converges in exactly $\alpha + \beta$ steps
    in $\M_{\alpha + \beta}$ and therefore terminates the proof of
    Lemma~\ref{lemma:existenceOfModel}.
  \end{proof}
  
  Now Theorem~\ref{thm:closuresum} immediately follows from 
    Lemmas~\ref{lemma:acceptablemodel},
    ~\ref{lemma:Sumacceptablemodel} and~\ref{lemma:existenceOfModel} when  applied to the formulas $\chi, \psi$ and $\Psi$ defined in~\ref{eq:defchi}, ~\ref{eq:defpsi} and~\ref{eq:defPsi} respectively.

    \medskip

    In the introduction we used $\OrdLmu$ to denote the set of closure
    ordinals of formulas of the modal $\mu$-calculus.  This section
    yields an insight on \Czarnecki's work \cite{Czarnecki} by proving
    the closure of $\OrdLmu$ under the ordinal sum.  The general
    problem of characterizing $\OrdLmu$ is open.  At the time of
    writing this paper, it is our opinion that still a few ordinals are
    known to belong to $\OrdLmu$---all of them can be
    constructed from the cardinals $1, \omega$ and $\omega_{1}$ by
    iterating the binary ordinal sum.
    Our results from Section~\ref{sec:continuousfragment} show that no
    other infinite regular cardinal $\kappa$ (apart from $\omega$ and
    $\omega_{1}$) can be proved to belong to $\OrdLmu$ in a
    straightforward way, that is, by relying on the \kcontinuity of
    some formula in $\Lmu$ and on the generalized Kleene theorem
    (Proposition~\ref{prop:convkcontinuous}). 
    Therefore, any other membership of $\OrdLmu$ requires a very
    different justification from the known ones.  
    New questions about $\OrdLmu$ need to be raised, such as whether
    this set is closed under other ordinal operations. Let us mention
    that a 
    recent work \cite{Milanese} exhibits a rich structure
    for
    closure ordinals of the modal $\mu$-calculus on bidirectional
    models.  It is conceivable that 
    studying closure ordinals on restricted classes of models will
    eventually yield a finer understanding of the structure of
    $\OrdLmu$.


\bibliographystyle{abbrv}
\bibliography{biblio}

\end{document}